\documentclass{article}
\usepackage{graphicx}
\usepackage{amsfonts}
\usepackage{amsmath}
\usepackage{amssymb}
\usepackage{url}
\usepackage{fancyhdr}
\usepackage{indentfirst}
\usepackage{epstopdf}
\usepackage{enumerate}
\usepackage{multirow}
\usepackage{dsfont}
\usepackage{caption}
\usepackage{subcaption}
\usepackage[colorlinks=true,citecolor=blue,breaklinks=true]{hyperref}
\usepackage{microtype}
\usepackage{amsthm}
\usepackage{color}
\usepackage{natbib}
\usepackage{comment}
\usepackage{bbm}

\def\d{\mathrm{d}}

\newcommand{\VaR}{\mathrm{VaR}}
\newcommand{\RVaR}{\mathrm{RVaR}}

\newcommand{\ES}{\mathrm{ES}}

\newcommand{\E}{\mathbb{E}}

\newcommand{\R}{\mathbb{R}}

\newcommand{\p}{\mathbb{P}}

\newcommand{\id}{\mathds{1}}

\renewcommand{\(}{\left(}
\renewcommand{\)}{\right)}

\renewcommand{\ge}{\geqslant}
\renewcommand{\le}{\leqslant}
\renewcommand{\geq}{\geqslant}
\renewcommand{\leq}{\leqslant}
\renewcommand{\epsilon}{\varepsilon}
\newcommand{\esssup}{\mathrm{ess\mbox{-}sup}}
\newcommand{\essinf}{\mathrm{ess\mbox{-}inf}}

\theoremstyle{plain}
\newtheorem{theorem}{Theorem}
\newtheorem{corollary}{Corollary}
\newtheorem{lemma}{Lemma}
\newtheorem{proposition}{Proposition}
\theoremstyle{definition}

\newtheorem{example}{Example}

\theoremstyle{remark}
\newtheorem{remark}{Remark}

\newcommand{\cet}{\begin{center}}
\newcommand{\ecet}{\end{center}}

\usepackage{setspace}

\begin{document}

\title{Risk Aggregation under Dependence Uncertainty and an Order Constraint}

\author{Yuyu Chen\thanks{Department of Statistics and Actuarial Science,
  University of Waterloo, Canada.
  \texttt{y937chen@uwaterloo.ca}.}
  \and Liyuan Lin\thanks{Department of Statistics and Actuarial Science,
  University of Waterloo, Canada.
  \texttt{l89lin@uwaterloo.ca}.}
  \and Ruodu Wang\thanks{Department of Statistics and Actuarial Science, University of Waterloo, Canada.   \texttt{wang@uwaterloo.ca.}}}

\maketitle

%

%
\begin{abstract}
We study the aggregation of two risks when the marginal distributions are known and the dependence structure is unknown, under the additional constraint that one risk is smaller than or equal to the other.
Risk aggregation problems with the order constraint are closely related to the recently introduced notion of the directional lower (DL) coupling.
The largest aggregate risk in concave order (thus, the smallest aggregate risk in convex order) is attained by the DL coupling. These results are further generalized to calculate the best-case and worst-case values of tail risk measures. In particular, we obtain analytical  formulas for bounds on Value-at-Risk. Our numerical results suggest that the new bounds on risk measures with the extra order constraint can greatly improve those with full dependence uncertainty.
\end{abstract}

\textbf{Keywords:} risk aggregation; risk measures; Value-at-Risk; concave order; directional lower coupling.

\section{Introduction}

Quantifying the risk of a portfolio has gained much interest in the literature of finance and actuarial science. To accurately estimate the risk level,  the joint distribution of the risks needs to be specified. However, it is challenging to estimate or test the dependence structure of a portfolio. Given known marginal distributions but unspecified dependence structure of risks, one of the most relevant problems is to find the worst-case (the largest possible) and the best-case (the smallest possible) values of a risk measure over all the possible dependence structures; see \cite{EP06}, \cite{BJW14} and \cite{EPR13, EWW15} for  general discussions.

While bounds for risk measures calculated based on the sole knowledge of marginal distributions are generally wide,  many attempts have been made to narrow them by incorporating partial dependence information into the problem. For instance, a variance constraint is imposed at the portfolio level by \cite{bernard2017value}. A lower bound is placed on the corresponding copula of risks by \cite{puccetti2016var}.  \cite{puccetti2017reduction} assumed that certain groups of risks are independent while the dependence structure is unknown within each group. \cite{bernard2017risk} considered a   partially specified factor model with dependence uncertainty.

 In the literature of isotonic regression, order constraint on the expectations of target variables has been widely used in many practical applications; see Section 1 of \cite{henzi2019isotonic} for an overview.  
 For two  random variables $\xi_{1}$ and $\xi_{2}$, an isotonic regression problem has the constraint $\E[\xi_{1}]\le \E[\xi_{2}]$.
  In many situations, while the risks $\xi_{1}$ and $\xi_{2}$ can be affected by a common shock $Z$ (e.g., market risk, pandemic, natural disaster), one can impose a stricter but natural assumption, that is, $\E[\xi_{1}|Z]\le \E[\xi_{2}|Z]$. In this paper, we study the aggregation $S=X+Y$ given known marginal distributions with the order constraint $X\le Y$, which might arise from, for instance, the above setting where  $X=\E[\xi_{1}|Z]$ and $Y=\E[\xi_{2}|Z]$.
  In practice, insurance companies can divide the loss of a portfolio into different categories according to the riskiness of the contract, and the order constraint naturally holds in situations where one risk triggers another. For instance, when floods occur, the higher floors of apartments/houses will suffer losses only if there is a huge damage in lower floors.  
As  another example, some cost categories for an insurance company, such as rehabilitation costs,  can only occur as a consequence of some severe disease.

  Before imposing the order constraint, one should verify that one of the  two distributions is stochastically smaller than the other. For real data, this relation can be tested via, e.g., the methods of \cite{barrett2003consistent}. Statistical inference for distributions ordered stochastically can be carried out  through the  isotonic distributional regression of \cite{henzi2019isotonic}.

  Fix an atomless probability space $(\Omega,\mathcal A,\p)$ and let $\mathcal M$ be the set of cdfs on $\R$. For  $F, G\in \mathcal{M}$  such that $F$ is stochastically smaller than $G$, define the set
$$
\mathcal F^o_2 (F,G) =\{(X,Y): X\sim F,~Y\sim G,~X\le Y\}.
$$
Here and throughout, the inequality $X\le Y$ is understood in the almost sure sense.
For a risk measure $\rho$, we are interested in
 the  worst-case and best-case values of $\rho$ over the set $\mathcal F^o_2 (F,G)$ denoted by
 \begin{equation}\label{pro:main}
\overline{\rho}(\mathcal F^o_{2}(F,G)):=\sup\{\rho(X+Y): (X,Y)\in\mathcal F^o_2 (F,G)\},
 \end{equation}
 and
  $$\underline{\rho}(\mathcal F^o_{2}(F,G)):=\inf\{\rho(X+Y): (X,Y)\in\mathcal F^o_2 (F,G)\}.$$

We mainly deal with the case where $\rho$ is a \emph{tail risk measure} introduced in \cite{LW20}.  The class of tail risk measures includes some of the most prevalent risk measures such as Value-at-Risk (VaR), Expected Shortfall (ES), and Range Value-at-Risk (RVaR). Generally speaking,  the value of a tail risk measure is determined by the risk's upper tail behavior. A key feature for a tail risk measure $\rho$ is that there exists another risk measure $\rho^{*}$, called the generator, such that $\rho(X) = \rho^*(X^{*})$ where the random variable $X^{*}$ follows the upper tail distribution of the random variable $X$.

In an unconstrained problem (i.e., only the marginal distributions of the two risks are known), for a tail risk measure $\rho$ such that $\rho^{*}$ is consistent with concave order,  the worst-case value of $\rho$ is attained by letting the upper tail risks be countermonotonic (i.e., the lower Fr{\'e}chet-Hoeffding bound). In particular, if $\rho$ is VaR, early results date back to  \cite{M81} and \cite{R82}.  Aggregation of more than two risks  is much more challenging; see \cite{WPY13}, \cite{PR13}, \cite{JHW16} and \cite{BLLW20} for some analytical results. The Rearrangement Algorithm (RA) is developed by  \cite{PR12} and \cite{EPR13} for numerical computation.

The problem with the order constraint is more sophisticated.
Recently, \cite{AMZ20} obtained the pointwise lower bound $D_*^{F,G}$ on joint distribution of $(X,Y)$ in $\mathcal F^o_{2}(F, G)$.
In mass transportation theory,
\cite{NW20} proposed the \emph{directional optimal transport} between two random variables, and
the corresponding joint distribution  is also $D_*^{F,G}$.
The transport is called {directional} because of the constraint $Y\ge X$;
that is,  $X$ can only be transported upwards to $Y$.
Since $D_*^{F,G}$ is the smallest distribution function among all joint distributions of $(X,Y)\in \mathcal F^o_{2}(F, G)$,
we will call the distribution $D_*^{F,G}$ the \emph{directional lower (DL) coupling},
and $(X^*,Y^*)\sim D_*^{F,G}$ is said to be \emph{DL-coupled}.
From the minimality of $D_*^{F,G}$ and results on concordance order of \cite{muller2000some},  $X^*+Y^*$ is the largest
in concave order among   $X+Y$ where $(X,Y)\in \mathcal F^o_{2}(F, G)$.

In general, we are interested in risk measures such as $\VaR$ and $\RVaR$, which are not monotone in concave or convex order. Therefore, the DL coupling does not give the maximum or minimum values of these risk measures, and  considerable new techniques need to be developed to find bounds on these risk measures.
Although $\VaR$ is not monotone in convex or concave order,
its generator, the essential infimum, is monotone in concave order.
As the main contribution of the paper (Theorem \ref{thm2}),
we show that
for a tail risk measure $\rho$ with a generator $\rho^{*}$ that is monotone in concave order (such as VaR and RVaR),  the solution to the constrained problem \eqref{pro:main} can be obtained by using the upper tail distributions of risks.
Moreover, the worst-case value of $\rho$ with the order constraint is attained by letting the two upper tail risks be DL-coupled.
The above assertions  on tail risk measures are based on a novel technical result of monotone embedding (Theorem \ref{lem:embed}).

  Despite its natural form, the order constraint in this paper can be quite strong and may not be easy to verify in some applications. Moreover, significant reduction of uncertainty bounds occurs when the two risks have comparable sizes, making the order constraint harder to justify; see Section \ref{sec6}. 
  As such, our contributions should be seen as mainly theoretical,  and they will serve as fundamental tools for applications emerging in the future.

The rest of the paper is organized as follows.
In Section \ref{sec2}, we give a brief review on comonotonicity, countermonotonicity, and the {DL coupling}. In Section \ref{sec3}, we study the worst-case dependence structures of risk aggregation with the order constraint in concave order. In Section \ref{sec4}, the notion of strong stochastic order is introduced. With this notion, we obtain several useful theoretical results.  The main technical contributions are contained in Section \ref{sec5}, where we obtain worst-case and best-case values of tail risk measures with the order constraint. Analytical results for VaR and probability bounds are obtained. In Section \ref{sec6}, numerical studies are conducted to illustrate the impact of the order constraint on the bounds of risk measures. Some concluding remarks and an open question are discussed in Section \ref{sec:con}.

We conclude this section by providing additional notations and terminologies that will be used throughout this paper. A cdf $F$ is said to be smaller than a cdf $G$ in stochastic order if $F\geq G$, denoted by $F\le_{\rm st} G$. Throughout, whenever $\mathcal F^o_2 (F,G)$ appears,  $F$ and $G$ are two distributions satisfying $F\le_{\rm st} G$. A cdf $F$ (or a random variable $X\sim F$) is said to be smaller than a cdf $G$  (or a random variable $Y\sim G$)  in concave order if  $\E[u(X)]\le \E[u(Y)]$ for all concave functions $u:\R\to \R$ provided that the expectations exist, and we denote this by $X\le_{\rm cv} Y$ or $F\le_{\rm cv} G$.
Further, $F$  is smaller than $G$  in convex order if $G \le_{\rm cv  } F$, and this is denoted by $F \le_{\rm cx  } G$. The order $X\le_{\rm cx} Y$ for two random variables $X$ and $Y$ is defined similarly. For more properties of these stochastic orders, we refer to \cite{SS07}.
A law-invariant risk measure $\rho$ is a mapping from $\mathcal M$ to $\R$.
In addition, we write $\rho(X)=\rho(F)$ for a random variable $X$ with distribution $F$; thus, $\rho$ can also be interpreted as a mapping from the set of random variables to $\R$.
 An empty set is denoted by $\emptyset$. By convention,  $\inf \emptyset=\infty$ and $\sup \emptyset =-\infty$.


\section{The directional lower coupling}\label{sec2}
In this section, we collect some basic results on comonotonicity, countermonotonicity, and the DL coupling, which will be useful for our paper.

A random vector $(X,Y)$ is said to be \emph{comonotonic} if there exists a random variable $U$ and two increasing functions $f$ and $g$ such that $X=f(U)$ and $Y=g(U)$ almost surely.  A random vector $(X,Y)$ is \emph{countermonotonic} if $(X,-Y)$ is comonotonic. We refer to \cite{DDGKV02, DDGKTV06} for a review on comonotonicity and \cite{PW15} for negative dependence concepts including countermonotonicity.

 Let  the random vectors $(X^{ct},Y^{ct})$, $(X,Y)$ and $(X^{c},Y^{c})$ be such that they have the same marginal distributions, $(X^{ct},Y^{ct})$  is   countermonotonic, and $(X^{c},Y^{c})$ is comonotonic. It is well known that
 \begin{align}\label{eq:RW1} X^{c}+Y^{c}\le_{\rm cv  } X+Y\le_{\rm cv  }X^{ct}+Y^{ct};\end{align}
  see e.g., \citet[Corollary 3.28]{R13}.\footnote{We choose to work mainly with concave order instead of convex order because a major target of this paper is to study VaR bounds, and the generator of VaR is increasing in concave-order; see Sections \ref{sec4} and \ref{sec5}. Nevertheless, since $\le_{\rm cv}$ is the same as $\ge_{\rm cx}$, all statements on concave order in this paper can be equivalently stated using convex order. Convex order is common in the literature of risk management, e.g., \cite{DDGK05}.}
  For $F,G\in \mathcal{M}$, let $X^{c},X^{ct}\sim F$ and $Y^{c},Y^{ct}\sim G$. Note that if $F\leq_{\rm st}G$,  then $X^{c}\leq Y^{c}$, which can be easily checked by choosing $U$ as a uniform random variable over $(0,1)$, and choosing $f$ and $g$ as left quantiles of $F$ and $G$, respectively. Hence, $\mathcal F^o_2 (F,G)$ contains comonotonic random vectors. However, $(X^{ct},Y^{ct})$ may violate the order constraint unless the essential supremum of $F$ is less than or equal to the essential infimum of $G$. Therefore, $(X^{ct},Y^{ct})$   may not be in  $\mathcal F^o_2 (F,G)$.

To find an alternative for countermonotonicity in $\mathcal F^o_2 (F,G)$,
we need to introduce the DL coupling, whose distribution function is obtained by \cite{AMZ20}. Below, we explain the DL coupling in the context of mass transport  following  \cite{NW20}, which is motivated by treatment effect analysis and causal inference (e.g., \cite{M97}).
A directional coupling of $F$ and $G$ is the joint distribution of a random vector in $\mathcal F^o_2 (F, G)$ and the DL coupling is the special case of a direction coupling which corresponds to the directional optimal transport of \cite{NW20}. Let $X\sim F$ and $Y\sim G$. Denote by $\mu_F$ and $\mu_G$ the Borel probability measures generated by $F$ and $G$, respectively.
The directional optimal transport from $X$ to $Y$ can be constructed by considering the common part and the singular parts of $\mu_{F}$ and $\mu_{G}$ separately. We first assume that $F$ and $G$ are continuous distributions.  
 The common part $\mu_{F}\wedge  \mu_{G}$ is defined as the maximal measure $\theta$ such that $\theta\le \mu_{F}$ and $\theta\le \mu_{G}$. 
 The singular parts of $\mu_{F}$ and $\mu_{G}$ are defined as $\mu'_{F}=\mu_{F}-\mu_{F}\wedge  \mu_{G}$ and $\mu'_{G}=\mu_{G}-\mu_{F}\wedge  \mu_{G}$. The shaded areas of density plots in Figure \ref{fig:density} illustrate the idea of the common   and singular parts for two Pareto distributions $F(x)=1-1/x$ for $x\ge 1$, and $G(y)=1-2/y$ for $y\geq 2$.
\begin{figure}[htbp]
\centering
\caption{Common   and singular parts of $\mu_{F}$ and $\mu_{G}$}\label{fig:density}
\includegraphics[height=4.5cm]{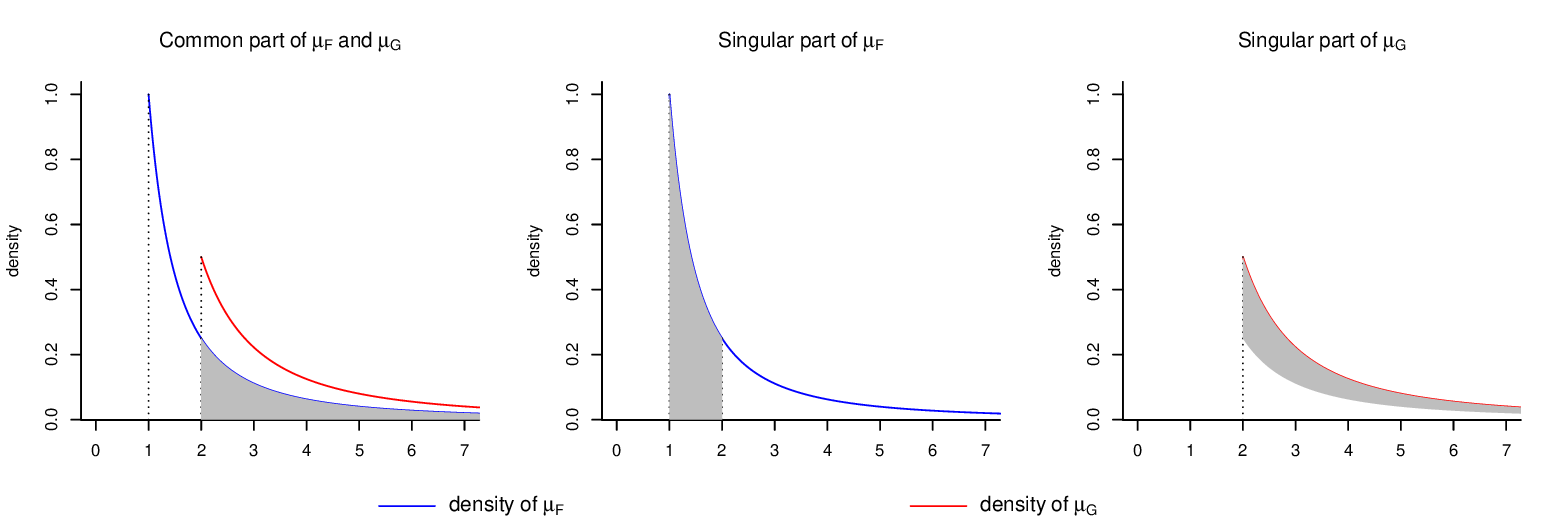}
\end{figure}

 The directional optimal transport  between $\mu_F$ and $\mu_G$ can be described in two pieces. First, the common part of $\mu_{F}$ and $\mu_{G}$ couples identically to each other. The transport from the singular part of $\mu_F$ to the singular part of  $\mu_G$, denoted by $T^{F, G}$, is defined as
$$T^{F,G}(x)=\inf\left\{z \ge x:F(z)-G(z)<F(x)-G(x)\right\}.$$
 Corollary 2.4 of \cite{NW20} gives the following representation of the DL coupling $D_*^{F,G}$, the joint distribution of $(X,Y)$ obtained above,
\begin{equation}\label{eq:rw1}
D_*^{F,G}(x,y)=\left\{
\begin{aligned}
&G(y)\quad&\text{if}~y \le x,\\
&F(x)-\inf_{z\in [x,y]}\{F(z)-G(z)\} \quad &\text{if}~y>x,
\end{aligned}
\right.
\end{equation}
which is also the bivariate distribution function in Theorem 6 of \cite{AMZ20}.
For a random vector $(X,Y)\sim D_*^{F,G}$, we say that $(X,Y)$ is \emph{DL-coupled}. Since DL coupling couples the common part of distributions to itself via the identity, it is a maximal coupling which maximizes $\p(X=Y)$ given the marginal distributions of $X$ and $Y$; see e.g., \citet[p.~104-112]{T00}. {DL coupling} differs from countermonotonicity in general, and they coincide if the essential supremum of $F$ is less than or equal to the essential infimum of $G$. In this special case,  all couplings between $F$ and $G$ are directional. In Figure \ref{copulas}, the support of the copula representing the DL coupling is plotted  for three pairs of Pareto distributions. Since $F$ and $G$ are identical in Figure \ref{c1}, the DL coupling is equivalent to comonotonicity. The DL coupling in Figure \ref{c2} is a simple combination of comonotonicity and countermonotonicity on the common part and singular parts of $\mu_{F}$ and $\mu_{G}$, respectively. The DL coupling in Figure \ref{c3} is more similar to countermonotonicity.
We warn the reader that, in general, the DL coupling can be much more complicated than these simple cases for other choices of marginal distributions. \cite{NW20} showed that the
DL coupling is the combination of one comonotonic coupling and countably many countermonotonic couplings, but these countermonotonic couplings may not be between conditional distributions on intervals like in these examples; see Proposition 2.6 and Example 6.3 of \cite{NW20}.
\begin{figure}[h]
     \caption{Support of the copula of $(U_{1},U_{2})=(F(X),G(Y))$ where $(X,Y)\sim D_*^{F,G}$}
        \label{copulas}
     \centering
     \begin{subfigure}[b]{0.3\textwidth}
         \centering
         \includegraphics[width=\textwidth]{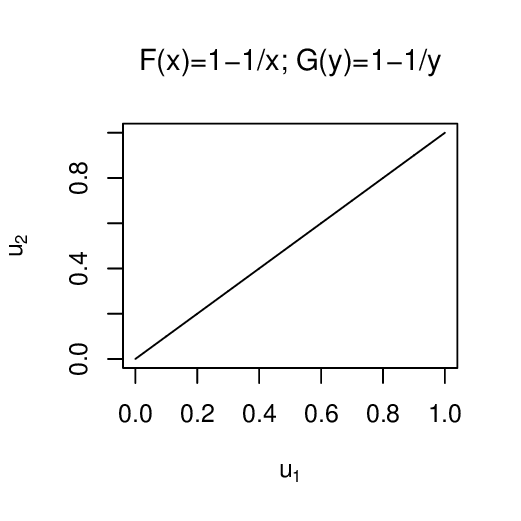}
         \caption{}
         \label{c1}
     \end{subfigure}
     \hfill
     \begin{subfigure}[b]{0.3\textwidth}
         \centering
         \includegraphics[width=\textwidth]{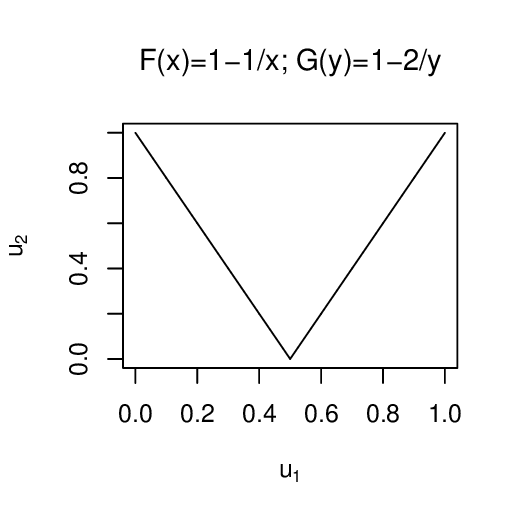}
         \caption{}
         \label{c2}
     \end{subfigure}
     \hfill
     \begin{subfigure}[b]{0.3\textwidth}
         \centering
         \includegraphics[width=\textwidth]{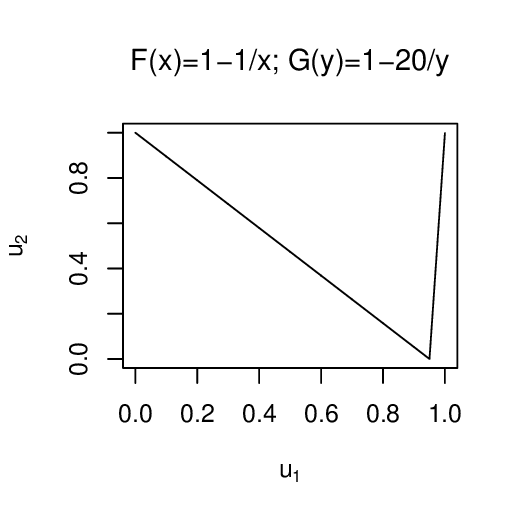}
         \caption{}
         \label{c3}
     \end{subfigure}
\end{figure}

One can also construct {DL coupling} for non-continuous distributions $F$ and $G$ satisfying $F\le_{\rm st}G$. The idea is first to convert the distributions $F$ and $G$ to continuous distributions $F_{c}$ and $G_{c}$ by a monotone transformation. Thereafter, construct DL coupling $D_*^{F_{c},G_{c}}$ and reverse the transformation back to non-continuous case; see Section 5.4 of \cite{NW20} for more details. As an important fact,  DL-coupled $(X,Y)$ always exists if $F\le_{\rm st}G$, and it has the distribution function \eqref{eq:rw1}.

Note that {DL coupling} may render the transport from $X$ to $Y$ randomized. That is, a realization of $G$ may have two pre-images through directional optimal transport.
In the language of mass transport theory, the directional optimal transport is Kantorovich-type but not necessarily Monge-type.
We use the following example to illustrate how {DL coupling} introduces such randomness and affects the aggregation of risks.
\begin{example}[Pareto risks: {DL coupling}]\label{ex:pareto}
Suppose that two risks follow Pareto distributions $F(x)=1-1/x$ for $x\ge 1$, and $G(y)=1-2/y$ for $y\geq 2$. Since $F\le_{\rm st}G$, we can take $(X,Y)\sim D_*^{F,G}$.
The directional optimal transport between the singular parts of $\mu_F$ and $\mu_G$ is
$$T^{F,G}(x)=\inf\left\{z \ge x:\frac{1}{z}<1-\frac{1}{x}\right\}=\frac{x}{x-1},~~~x\in(1,2].$$
If $x=1$, $T^{F,G}(x)=\inf\emptyset=\infty$. The directional optimal transport between the common part of $\mu_F$ and $\mu_G$ is the identical transport. Thus, for a real number $y \ge 2$, its pre-image is either $y$ through the identical transport or $y/(y-1)$ through $T^{F,G}$.  Let $c^{*}=(c-\sqrt{c^{2}-4c})/2$ for $c\in [4,\infty]$. By Corollary 2.9 of \cite{NW20}, we have
\begin{align*}
\p(X+Y\le c)&=(\mu_{F}\wedge \mu_{G})\left(\left[-\infty,\frac{c}{2}\right]\right)+\mu'_{F}\left(\left[c^{*},2\right]\right)\\
&=\left(F\left(\frac{c}{2}\right)-F\left(2\right)\right)+\left(F(2)-F\left(c^{*}\right)\right)\\
&=\frac{c+\sqrt{c^2-4c}-4}{2c}.
\end{align*}
This example will be continued in Examples \ref{ex:essinf}, \ref{ex:varbounds} and \ref{ex:prob}.
In particular, if the order constraint is imposed, the DL coupling  leads to the  largest essential infimum of $X+Y$, but it
does not lead to the largest or the smallest probability $\p(X+Y\le t)$ in general.
\end{example}

\section{Optimality of the directional lower coupling}\label{sec3}

In this section, we study the optimal dependence structures of $(X, Y)\in\mathcal F^o_2 (F, G)$ in the sense of concave order, or equivalently, convex order.
As we have seen in \eqref{eq:RW1},  comonotonicity of   $(X,Y)$ yields the smallest $X+Y$ in concave order among all possible dependence structures with given marginal distributions, and hence it also yields  the smallest concave order of $X+Y$ for $(X,Y)\in  \mathcal{F}^o_2 (F,G)$.
On the other hand, a simple result from \cite{NW20} shows that
 the DL coupling of   $(X,Y)$ yields the  largest concave order of $X+Y$ over $\mathcal{F}^o_2 (F,G)$. The concave ordering bounds are very useful for the calculation of   bounds on risk measures.
\begin{lemma}\label{lem:opt}
  For $(X,Y),(X^c,Y^c),(X',Y')\in \mathcal F^o_2(F,G)$
such that $(X^c,Y^c)$ is comonotonic and $(X',Y')$ is DL-coupled, we have
$$ X^c+Y^c \le_{\rm cv}  X+Y  \le_{\rm cv}  X'+Y' .$$
\end{lemma}

\begin{proof}
  The first inequality can be found in, e.g., Theorem 3.5 of \cite{R13}. For the second inequality, by Theorem 2.2 (i) of \cite{NW20}, $(X',Y')$ is the smallest element of $\mathcal F^o_2(F,G)$ in concordance order. Equivalently, $(X',Y')$ is the smallest element of $\mathcal F^o_2(F,G)$ in supermodular order (e.g., Theorem 2.5 of \cite{muller2000some}).
  It is well known that the function $ (x,y)\mapsto -u(x+y)$ on $\R^2$ for any concave function $u:\R\to \R$
 is supermodular.
Hence $\E[u(X+Y)]\le \E[u(X'+Y')]$, and  therefore $ X+Y \le_{\rm cv}  X'+Y'$.
\end{proof}
Next, we use the above concave ordering bounds to obtain bounds on risk measures. We refer to \cite{FS16} for an overview on risk measures. For a risk measure $\rho:\mathcal{M}\rightarrow \R$, we define three commonly used properties:
\begin{enumerate}[(i)]
\item A risk measure $\rho$ is \emph{monotone} if $\rho(F)\le \rho(G)$ whenever $F\le_{\rm st} G$;
\item A risk measure $\rho$ is \emph{$\le_{\rm cv}$-consistent} if $\rho(F)\le \rho(G)$ whenever $F\le_{\rm cv} G$;
\item A risk measure $\rho$ is \emph{$\le_{\rm cx}$-consistent} if $\rho(F)\le \rho(G)$ whenever $F\le_{\rm cx} G$.
\end{enumerate}

 Many popular risk measures are monotone, such as Value-at-Risk (VaR), Expected Shortfall (ES), and Range-VaR (RVaR).    For $F\in \mathcal{M}$ and $q\in (0,1]$, the left VaR denoted by VaR$^{L}_{q}:\mathcal{M}\rightarrow \R$ is given by
$$\VaR^{L}_{q}(F)=F^{-1}(q)=\inf\{t\in\R:F(t)\geq q\}.$$
For  $p\in [0,1)$, the right VaR denoted by VaR$^{R}_{p}:\mathcal{M}\rightarrow \R$ is given by
$$\VaR^{R}_{p}(F)=F^{-1}(p+)=\inf\{t\in\R:F(t)> p\}.$$
 For $p=0$ and $q=1$, $\VaR^{R}_{0}(F)$ and $\VaR^{L}_{1}(F)$ correspond to the essential infimum and essential supremum of $F$ which are also denoted by $\essinf (F)=F^{-1}(0)$ and $\esssup (F)=F^{-1}(1)$, respectively.
 For $F\in \mathcal M$, ES$_{p}:\mathcal{M}\rightarrow \R$ for $p\in(0,1)$ is defined as
$$\ES_{p}(F)=\frac{1}{1-p}\int_{p}^{1}\VaR^{R}_{u}(F)\mathrm{d}u,$$
and
 $\RVaR_{p,q}:\mathcal{M}\rightarrow \R$ for $0\leq p<q<1$ is defined as
 $$\RVaR_{p,q}(F)=\frac{1}{q-p}\int_{p}^{q}\VaR^{R}_{u}(F)\d u.$$
 The class of RVaR is proposed by \cite{cont2010robustness} as robust risk measures; see \cite{ELW18} for its properties.

For $p\in (0,1)$, $\VaR^{L}_{p}$ and $\VaR^{R}_{p}$ are neither $\le_{\rm cx}$-consistent nor $\le_{\rm cv}$-consistent. On the other hand, $\ES_p$ and VaR$^{L}_{1}$ are $\le_{\rm cx}$-consistent, and $\RVaR_{0,q}$ and  VaR$^{R}_{0}$ are $\le_{\rm cv}$-consistent. Monetary risk measures (see \cite{FS16}) that are $\le_{\rm cx}$-consistent are characterized by \cite{MW20} and they admit an ES-based representation. Using  Lemma \ref{lem:opt}, we immediately obtain the following  bounds of $\le_{\rm cv}$-consistent and $\le_{\rm cx}$-consistent risk measures. Recall that the worst-case and best-case risk measures are defined as, respectively,
 $$
\overline{\rho}(\mathcal F^o_{2}(F,G))=\sup\{\rho(X+Y): (X,Y)\in\mathcal F^o_2 (F,G)\},
 $$
 and
  $$\underline{\rho}(\mathcal F^o_{2}(F,G))=\inf\{\rho(X+Y): (X,Y)\in\mathcal F^o_2 (F,G)\}.$$

\begin{corollary}\label{rm:cv}
 Suppose that $(X,Y),(X^c,Y^c),(X',Y')\in \mathcal F^o_2(F,G)$ such that $(X^c,Y^c)$ is comonotonic and $(X',Y')$ is DL-coupled. If $\rho $ is $\le_{\rm cv}$-consistent, then
$$\underline{\rho}(\mathcal F^o_{2}(F,G))= \rho(X^c+Y^c) \le \rho(X+Y)  \le  \rho(X'+Y') =\overline{\rho}(\mathcal F^o_{2}(F,G)).$$
If $\rho $ is $\le_{\rm cx}$-consistent, then
$$\underline{\rho}(\mathcal F^o_{2}(F,G))= \rho(X'+Y') \le \rho(X+Y)  \le  \rho(X^c+Y^c) =\overline{\rho}(\mathcal F^o_{2}(F,G)).$$
\end{corollary}

 As   seen from Corollary \ref{rm:cv}, for a $\le_{\rm cx}$-consistent risk measure, such as a law-invariant convex or coherent risk measure, the extra order constraint does not improve the worst-case risk value obtained under comonotonicity, as in the case without the order constraint.

Since essential infimum and essential supremum are $\le_{\rm cv}$-consistent and $\le_{\rm cx}$-consistent respectively, with Corollary \ref{rm:cv}, we give analytical results on the worst-case value of essential infimum and the best-case value of essential supremum in the following theorem. This result will be used to derive the worst-case and best-case values of VaR in Section \ref{sec5}.

\begin{theorem} \label{essinf}
For continuous distributions $F$, $G$ and $(X,Y)\sim D^{F,G}_*$, we have
\begin{equation}\label{eq:essinf}
\overline{\essinf}(\mathcal F^o_{2}(F,G))=\essinf(X+Y)
=\min\left\{\inf_{x\in[F^{-1}(0),G^{-1}(0)]}\left\{T^{F,G}(x)+x\right\},2G^{-1}(0)\right\},
\end{equation}
and
\begin{equation}\label{eq:esssup} \underline{\esssup}(\mathcal F^o_{2}(F,G))=\esssup(X+Y)
=\max\left\{\sup_{x\in[F^{-1}(1), G^{-1}(1)]}\left\{\hat T^{F,G}(x)+x \right\},2F^{-1}(1)\right\},\end{equation}
where $\hat T^{F,G}(x)=\sup\{t\le x: F(t)- G(t)< F(x)-G(x)\}$.
\end{theorem}

\begin{proof}
We first prove the statement \eqref{eq:essinf} on the worst-case value of  the essential infimum. Combining Corollary \ref{rm:cv} and the fact that essential infimum is $\le_{\rm cv}$-consistent, we have the first equality in  \eqref{eq:essinf}.
To prove the second equality in  \eqref{eq:essinf},  for $(X,Y)\sim D^{F,G}_*$, let
$$
t^{*}:=\sup\left\{t\in\mathbb{R}:\text{ for all } (x,y)\in\mathbb{R}^{2}\text{ such that } t=x+y\text{~and~} D(x,y)=0\right\},
 $$
where $D=D^{F,G}_*$. It is straightforward to see  from the definition that $ t^{*} = \essinf(X+Y).$
Therefore, we have
\begin{equation}\label{bd}
F^{-1}(0)+G^{-1}(0)\le \essinf(X+Y) \le 2G^{-1}(0).
\end{equation}
If $F^{-1}(0)=G^{-1}(0)$,  then clearly $ \essinf(X+Y) = 2G^{-1}(0)$. For the rest of the proof, we assume   $F^{-1}(0)<G^{-1}(0)$.
\begin{enumerate}[(i)]
\item
If $x\le F^{-1}(0)$ or $y\le G^{-1}(0)$, $D(x,y)=0$.
\item
If $x> G^{-1}(0)$ and $y> G^{-1}(0)$, it is easy to check $D(x,y)>0$.
\item
 If $F^{-1}(0) < x\le G^{-1}(0)<y$, we have
 $$D(x,y)=\max\left\{-\inf_{z\in [G^{-1}(0),y]} \{F(z)-G(z)-F(x)\},0\right\}.$$
 By definition of $T^{F,G}$, if $F^{-1}(0) < x\le G^{-1}(0)<y\leq T^{F,G}(x)$, $D(x,y)=0$. On the other hand, if $F^{-1}(0) < x\le G^{-1}(0)\le T^{F,G}(x)<y $, $D(x,y)>0$.

\end{enumerate}
As a result, $D(x,y)=0$ if and only if one of the following holds: $x\le F^{-1}(0)$, $y\le G^{-1}(0)$ or $F^{-1}(0) < x\le G^{-1}(0)<y\leq T^{F,G}(x)$.
 Let
 $$s:=\min\left\{\inf_{x\in[F^{-1}(0),G^{-1}(0)]}\left\{T^{F,G}(x)+x\right\},2G^{-1}(0)\right\}.$$
 We will show   $t^{*}=s$. For $x,y\in \R$, suppose that $x+y=s$. If $F^{-1}(0)<x\le G^{-1}(0)$, we have $y\leq T^{F,G}(x)$. If $x > G^{-1}(0)$, we have $y < G^{-1}(0)$. That is, for any $x,y\in \R$ such that $x+y=s$, $D(x,y)=0$. Thus we have $t^{*}\ge s$.  For any $g,h\in \R$, suppose that $g+h=t^{*}$. Then for any $\epsilon>0$, we have $D(g,h-\epsilon)=0$. Therefore, if $F^{-1}(0)<g\le G^{-1}(0)$, we have $h\leq T^{F,G}(g)+\epsilon.$ By letting $\epsilon$ goes to 0, we have $t^{*}=g+h\leq T^{F,G}(g)+g$ for any $g$ in $(F^{-1}(0), G^{-1}(0)]$. As $T(F^{-1}(0))=\infty$, we have  $t^{*}\le \inf_{x\in[F^{-1}(0),G^{-1}(0)]}\left\{T^{F,G}(x)+x\right\}$. By \eqref{bd}, $t^{*}\le 2G^{-1}(0)$. Thus we have $t^{*}\le s$, and the statement \eqref{eq:essinf} on the worst-case value of  the essential infimum holds.

Next, we show the statement \eqref{eq:esssup} on the best-case value of the essential supremum.  Let $\hat F(t):=1-F(-t)$, $\hat G(t):=1-G(-t)$ for $t \in \R$ and $\hat T^{F,G}(x)=\sup\{t\le x: F(t)- G(t)< F(x)-G(x)\}$ for $x\in \R$. $\hat F$ and $\hat G$ are the distributions of $-X$ and $-Y$. Then we have $-T^{\hat G,\hat F}(x)=\hat T^{F,G}(-x)$ for $x\in \R$. Note that
$$\underline\esssup(\mathcal F^{o}_2(F,G))=-\sup\{\essinf(-X-Y):(-Y,-X) \in \mathcal F^{o}_2(\hat G,\hat F)\}.$$
Applying \eqref{eq:essinf}, we get the desired equality.
\end{proof}
\begin{remark}
In the unconstrained case,  the worst-case value of the essential infimum and the best-case value of the essential supremum are attained  by countermonotonicity, i.e.,
$$\sup\{\essinf(X+Y): X\sim F,~Y\sim G\}=\inf_{x\in[0,1]}\left\{F^{-1}(x)+G^{-1}(1-x)\right\},  $$
$$\inf\{\esssup(X+Y): X\sim F,~Y\sim G\}= \sup_{x\in[0,1]}\left\{F^{-1}(x)+G^{-1}(1-x)\right\}. $$
Generally, these bounds are different from the bounds in Theorem \ref{essinf}.
\end{remark}

\begin{example}[Pareto risks: Essential infimum] \label{ex:essinf}
 In Example \ref{ex:pareto}, we derive the cdf of $X+Y$ for $(X,Y)\sim D^{F,G}_*$ where $F$ and $G$ are two Pareto distributions. We have
$$\p(X+Y\le c)=\frac{c+\sqrt{c^2-4c}-4}{2c}, ~~~~~~c\in[4,\infty).
 $$
  By Corollary \ref{rm:cv}, $\overline{\essinf}(\mathcal F^o_{2}(F,G))=\essinf(X+Y)=4$.  Alternatively, we can use the analytical result of Theorem \ref{essinf}.  As $T^{F,G}(x)=x/(x-1)$ for $x\in(1,2]$ and $T^{F,G}(1)=\infty$, we have $\inf_{x\in [1,2]}\left\{T^{F,G}(x)+x\right\}=4$. Thus, we have
 $\overline{\essinf}(\mathcal F^o_{2}(F,G))=\min\left\{4,4\right\}=4$.
  The unconstrained upper bound for $\essinf$ is $3+2\sqrt{2}$ which is larger than $\overline{\essinf}(\mathcal F^o_{2}(F,G))$.
  Both the constrained and unconstrained lower bounds of $\essinf$ are attained when $(X,Y)$ are comonotonic and we have $\underline{\essinf}(\mathcal F^o_{2}(F,G))=3.$
\end{example}

\section{Strong stochastic order and monotone embedding}\label{sec4}
In this section, we introduce the notion of {strong stochastic order}, and obtain several theoretical properties.
The new notion is crucial for the main results of this paper in Section \ref{sec5}.

For $F$, $G\in \mathcal M$, we say $F$ is smaller than $G$ in \emph{strong stochastic order} if $G(y)-G(x)\ge F(y)-F(x)$ for all $y\ge x\ge G^{-1}(0)$, denoted by $F\le_{\rm ss} G$. Equivalently, the function $x\mapsto G(x)-F(x)$ is decreasing  for $x\ge G^{-1}(0)$.
Note that the order $\le_{\rm ss}$ is stronger than $\le_{\rm st}$, and hence the name.
Intuitively, $G$ has more probability in any interval $(x,y]$ than $F$ if $x\ge G^{-1}(0)$.
 If $F$ and $G$ have densities $g$ and $f$ with respect to a dominating measure,  then  $F\le_{\rm ss} G$ if  and only if $g(x)\ge f(x)$ for $x\geq G^{-1}(0)$.
As far as we know, this notion of stochastic order is new to the literature.
We first provide some simple properties of the order $\le_{\rm ss}$.
\begin{proposition}\label{prop:1}
The strong stochastic order satisfies the following properties:
\begin{enumerate}[(i)]
\item
If $F\le_{\rm ss} G$ then $F\le_{\rm st} G$;
\item
Assuming $F^{-1}(0)=G^{-1}(0)$, $F\le_{\rm ss} G$ if and only if $F=G$;
\item
If $G^{-1}(0)=-\infty$, then $F\le_{\rm ss} G$ means $F=G$;
\item
The relation $\le_{\rm ss}$ is a partial order.
\end{enumerate}
\end{proposition}
\begin{proof}
\begin{enumerate}[(i)]
\item
By letting $y \to \infty$, we have
$1-G(x)\ge 1-F(x)$ for all $x \ge G^{-1}(0)$.
Hence, $G(x) \le F(x)$ for all $x \ge  G^{-1}(0)$. Moreover, $G(x)=0$ for $x< G^{-1}(0)$. Hence, $G(x)\le F(x)$ for all $x\in \R$, which gives $F\le_{\rm st} G$.
\item The
``$\Leftarrow$'' direction is obvious. For the ``$\Rightarrow$'' direction, by letting $x= F^{-1}(0)= G^{-1}(0)$, we have $F(x)- G(x)=0$. Thus, for all $y\ge x= F^{-1}(0)= G^{-1}(0)$, $F(y)\le G(y)$, which means $G\le_{\rm st} F$. Together with $F\le_{\rm st} G$ from (i), we have $F=G$.
\item
By (i), we have $ F^{-1}(0) \le  G^{-1}(0)=-\infty$. Thus, $ F^{-1}(0)= G^{-1}(0)=-\infty$. Hence, $F=G$ by (ii).
\item
Reflexivity is obvious and antisymmetry is implied by (i). Suppose that $F\le_{\rm ss} G$ and $G\le_{\rm ss} H$. By (i), $G\le_{\rm st} H$ and $\max\{ G^{-1}(0),  H^{-1}(0) \}= H^{-1}(0)$. We have $H(y)-H(x)\ge F(y)-F(x)$ for all $y\ge x\ge \max\{ G^{-1}(0),  H^{-1}(0) \}= H^{-1}(0)$. Transitivity of the order $\le_{\rm ss}$ is proved.\qedhere
\end{enumerate}
\end{proof}

Next, we discuss the problem of \emph{monotone embedding}, which is an important issue in the analysis of risk aggregation for tail risk measures in Section \ref{sec5}.
The problem is formulated as follows.
Suppose that $F \le_{\rm st} F' \le_{\rm st} G$ and  $(X,Y)\in \mathcal F^o_2(F,G)$,
and the question is whether there exists $X'\sim F'$ such that $X\le X'\le Y$ holds (in the almost sure sense).
The existence of such $X'$ is crucial to prove that we can use tail distribution to obtain bounds on tail risk measures (see Theorem \ref{thm2} below).
Unfortunately, in general, such $X'$ does not exist, even if we further assume that $(X,Y)$ is DL-coupled.

\begin{example}\label{ex:monemb}
Let $G$ be the Bernoulli($1/2$) distribution.
Take   $Y\sim G$, let $X=-Y$, and $F$ be the distribution of $X$.
Clearly, $(X,Y)$ is countermonotonic, and hence $(X,Y)  $ is DL-coupled.
Take another random variable $X'\sim  F'=\mathrm U[-1,1]$.
It is easy to see that $F\le_{\rm st} F' \le_{\rm st} G$.
Since $\p(X=Y)=1/2$ but $\p(X'=Y)=0$,  we know that $X\le X'\le Y$ cannot hold for any $X'\sim F'$.
\end{example}
The next theorem is the most important technical result which allows us to study the DL coupling using the strong stochastic order. The result says that, although  $F \le_{\rm st} F' \le_{\rm st} G$ is not sufficient for the existence of $X'$   in Example \ref{ex:monemb},
assuming the stronger relation $F \le_{\rm ss} F'$ would suffice.
\begin{theorem}[Monotone embedding] \label{lem:embed}
Suppose that $F \le_{\rm ss} F' \le_{\rm st} G$,
and
  $(X,Y)\sim D_*^{F,G}$.
  Then there exists $X'\sim F'$ such that $X\le X'\le Y$ almost surely
  and
    $(X',Y) $ is DL-coupled.
\end{theorem}

\begin{proof}
We first consider continuous distributions $F, F'$ and $G$.
Without loss of generality, we assume $\mu_{F}$ and $\mu_{G}$ are not mutually singular. As $(X,Y)\sim D_*^{F,G}$, the common part $\mu_{F}\wedge\mu_{G}$ of $\mu_{F}$ and $\mu_{G}$ are identically coupled. The singular part of $\mu_{F}$ is transported to the singular part of $\mu_{G}$ through $T^{F,G}$. Let $P$ be a joint distribution on $\R^{3}$ with marginals $P\circ X^{-1}=\mu_{F}$,  $P\circ (X')^{-1}=\mu_{F'}$ and $P\circ Y^{-1}=\mu_{G}$ such that $(X,Y)\sim D_*^{F,G}$. We will construct $P$ such that $(X',Y) $ is DL-coupled and $X\le X'\le Y$ almost surely.
  \begin{enumerate}[(i)]
  \item
  Let $\theta=\mu_{F}\wedge\mu_{G}$ and $\theta'=\mu_{F'}\wedge\mu_{G}$. As $F\le_{\rm ss}F'$, $\mu_F(a,b]\leq \mu_{F'}(a,b]$ for all $b\geq a\geq  (F')^{-1}(0)$. Thus the common part of $\mu_{F'}$ and $\mu_{G}$ covers the common part of $\mu_{F}$ and $\mu_{G}$, i.e., $\theta\wedge\theta'=\theta$. Therefore, we can always construct $P$ such that the measure $\theta$ of $\mu_{F}$, $\mu_{F'}$ and $\mu_{G}$ identically couples with each other. By further letting $P$ couple the measure $\theta'-\theta$ of $\mu_{F'}$ and $\mu_{G}$ identically, the common part $\theta'$ of $\mu_{F'}$ and $\mu_{G}$ identically couples.

  \item
     Next we focus on the directional optimal transports on the singular parts of distributions, i.e., $T^{F,G}$ and $T^{F',G}$. Let $P$ transport the singular part of $\mu_{F'}$ to the singular part of $\mu_{G}$ through $T^{F',G}$. Take $x$, $x'$ and $y$ satisfying $y=T^{F,G}(x)=T^{F',G}(x')$. Note that we will not consider the sets of $x$ and $x'$ such that $x=y$ or $x'=y$ as we are studying the singular parts of distributions. Thus we have $x<y$ and $x'<y$. A key property of  $T^{F,G}$
     is   that $F(z) - G(z) = F\left(T^{F,G}(z)\right)- G\left(T^{F,G}(z)\right)$ holds for all $z\in \R$; see Lemma 5.2 of \cite{NW20}.
     With this property and $F\le_{\rm ss} F'$, we have
     \begin{align*}
     F'(x)-G(x) & =F'(x)-F(x)+F(x)-G(x) \\& \le F'(y)-F(y)+F(y)-G(y)=F'(x')-G(x').
     \end{align*}
  Assume that $x'<x$. If $F'(x)-G(x)< F'(x')-G(x')$, as $x'<x< y=T^{F',G}(x')$, by definition of $T^{F',G}$, we have $x=y$ as a contradiction to $x<y$. If $F'(x)-G(x)= F'(x')-G(x')$, as $x'<x<y=T^{F',G}(x')$, $x$ is neither a point of strict increase nor a point of strict decrease of $F'-G$ in the sense of \cite{NW20}.
  By Proposition 5.1 of \cite{NW20}, the set of points which are neither of strict increase nor of strict decrease is a null set.
\end{enumerate}

  By (i) and (ii), we construct $P$ such that $(X',Y)\sim D_*^{F',G}$  and $X\le X'\le Y$ almost surely. Note that $(X,Y)\sim D_*^{F,G}$ and $(X',Y)\sim D_*^{F',G}$ do not necessarily imply $X\le X'\le Y$ almost surely due to the randomness of DL coupling which is illustrated by Example \ref{ex:pareto}. Therefore, the construction of $P$ in (i) is necessary. Next, we proceed to complete the proof for non-continuous distributions $F, F'$, and $G$. As the construction in (i) can also be applied to the common part of non-continuous distributions, we focus on the singular parts of distributions and assume that $\mu_{F}\wedge \mu_{G}=0$ and $\mu_{F'}\wedge \mu_{G}=0$. Let
  $$j(x)=x+\sum_{y\le x}\left|H(y)-H(y-)+(F(y)-F(y-))\id_{\left\{y<(F')^{-1}(0)\right\}}\right|,~~~x\in\R,$$
  where $H=F'-G$. The function $j$ is the summation of an identity function, the jumps of $H$ and the jumps of $F(x)$ for $x<(F')^{-1}(0)$. Denote by $j^{-1}:j(\R)\rightarrow \R$ the right-continuous inverse function of $j$. Let
$$J_{x}=[j(x-),j(x)]$$
  be the interval representing the jump of $j$ at $x$. If there is no jump at $x$, $J_{x}$ is a singleton. Next we convert the measure $\mu_{F'}$ to an auxiliary measure $\mu_{{F}_{c}'}$ with continuous cdf $F'_{c}$. We set $F'_{c}(z)=F'(j^{-1}(z))$ for $z\in j(\R)$. On the complement of $j(\R)$, $F'_{c}$ is defined by linearly interpolating from its values on $j(\R)$. In other words, if $\mu_{F'}$ has a jump at $x$, $\mu_{{F}_{c}'}$ is uniformly distributed on the interval $J_{x}$ with probability $\mu_{{F}_{c}'}(J_{x})=\mu_{F'}(\{x\})$. The auxiliary measures $\mu_{{F}_{c}}$ and $\mu_{{G}_{c}}$ with cdfs $F_{c}$ and $G_{c}$ can be constructed similarly from $\mu_{F}$ and $\mu_{G}$. The transformation implies that $G_{c}$ is also continuous and $F_{c}' \le_{\rm st} G_{c}$. Note that as $F \le_{\rm ss} F'$, for $x\geq (F')^{-1}(0)$, we have
  \begin{equation*}
  F(x)-F(x{-})\le F'(x)-F'(x{-}).
  \end{equation*}
  The above inequality implies that for any $x\geq (F')^{-1}(0)$, whenever $F(x)$ has a jump, $F'(x)$ must have one. Therefore, the transformation from $\mu_{F}$ to $\mu_{{F}_{c}}$ reduces all the atoms of $\mu_{F}$ and $F_{c}$ is continuous. Moreover, as $F \le_{\rm ss} F'$, the transformation ensures that $F_{c} \le_{\rm ss} F_{c}'$. Consequently, the orders on $F$, $F'$ and $G$ are preserved after the transformation and we have  $F_{c} \le_{\rm ss} F_{c}' \le_{\rm st} G_{c}$.

   Apply the result for continuous distributions on $F_{c}$, $F'_{c}$ and $G_{c}$ 
   and convert the transformation back to non-continuous distributions. As all transformations are monotone (see Theorem 5.5 of \cite{NW20}), the order $X\le X'\le Y$ still holds almost surely for non-continuous distributions $F$, $F'$ and $G$.
\end{proof}

In what follows, for any set $A \in \mathcal A$ with positive probability, let $H _{X|A}$ be the conditional distribution of $X$ given $ A$.
Moreover, $F^{[p,1]}$ is the upper $p$-tail distribution of $F$, namely
$$ F^{[p,1]}(x )=\frac{  (F(x)-p)_+ }{1-p},~~~x\in \R,$$
and $F^{[0,p]}$ is the lower $p$-tail distribution of $F$, namely
$$ F^{[0,p]}(x )=\frac{  F(x) \wedge p}{p},~~~x\in \R.$$
In other words, $F^{[p,1]}$ is the distribution of $F^{-1}(U)$ where $U\sim \mathrm{U}[p,1]$, and $F^{[0,p]}$ is the distribution of $F^{-1}(U)$ where $U\sim \mathrm{U}[0,p]$. The next proposition shows that the largest conditional distribution $H _{X|A}$ for $A \in \mathcal A$ with probability $1-p$ in strong stochastic order is the upper $p$-tail distribution $F^{[p,1]}$ where $X\sim F$. The event $A$ such that $H _{X|A}=F^{[p,1]}$ is called a $p$-tail event in \cite{WZ20}, which will be formally defined in Section \ref{sec5}.
\begin{proposition}\label{prop:2}
For $p\in (0,1)$,   any set $A\in \mathcal A$ of probability $1-p$ and  $X\sim F$,   $H _ {X|A}\le_{\rm ss} F^{[p,1]}$.
\end{proposition}
\begin{proof}  For any interval $[x,y]$ with $x\ge \left(F^{[p,1]}\right)^{-1}(0)$, we have
\begin{align*}
F^{[p,1]}(y)-F^{[p,1]}(x) = \frac{(F(y)-p)_+-(F(x)-p)_+ }{1-p}= \frac{ F(y) - F(x)  }{1-p},
\end{align*}
and
\begin{align*}
H _{X|A}(y)-H _{X|A}(x) & =  \p(x< X \le y  \mid    A) =  \frac{\p(\{x< X\le y \}\cap A)}{1-p}
\le  \frac{ F(y) - F(x)  }{1-p}.
\end{align*}
Hence, we have $H _{X|A} \le_{\rm ss} F^{[p,1]}$.
\end{proof}
Combining Theorem \ref{lem:embed} and Proposition \ref{prop:2}, we immediately arrive at the following corollary. This corollary will be used to establish the main result on the worst-case value of tail risk measures with the order constraint in Section \ref{sec5}.
\begin{corollary} \label{coro:1}
Let $A\in \mathcal A$ with probability $1-p$ and $p\in (0,1)$.
Suppose that $F \le_{\rm st} G$, $X\sim F$
and
  $(X_A,Y)\sim D_*^{H _{X|A},G^{[p,1]}}$.
  Then there exists $X'\sim F^{[p,1]}$ such that $X_A\le X'\le Y$ almost surely
  and
    $(X',Y) $ is DL-coupled.
\end{corollary}
\begin{proof}
Note that $F^{[p,1]}\le _{\rm st} G^{[p,1]}$ follows from $F \le _{\rm st} G$,
and $H _{X|A}\le_{\rm ss} F^{[p,1]}$ follows from Proposition \ref{prop:2}.
Applying Theorem \ref{lem:embed} with the condition $H _{X|A} \le_{\rm ss} F^{[p,1]}\le _{\rm st} G^{[p,1]}$ gives the desired result.
\end{proof}
%

\section{Risk measure and probability bounds}\label{sec5}

\subsection{Bounds on tail risk measures}
Evaluating the ``tail risk '', or the behavior of a risk beyond a high level, has become crucial in the regulatory frameworks for banking and insurance. To better understand the tail risk, \cite{LW20} provided an axiomatic framework of risk measures which can quantify the tail risk. Those risk measures are referred to as tail risk measures. This section is dedicated to studying the worst-case value of tail risk measures with the order constraint.

 For $p\in (0,1)$, a risk measure $\rho$ is a \emph{$p$-tail risk measure} if $\rho(F)=\rho(G)$ for all $F,G\in \mathcal{M}$ such that  $F^{[p,1]}= G^{[p,1]}$.  In other words, the value of a $p$-tail risk measure of random variable $X$ is determined by its distribution beyond $F^{-1}(p)$. The class of tail risk measures includes the most important regulatory risk measures VaR and ES, and those popular in the literature, such as RVaR and Gini Shortfall (\cite{furman2017gini}).

For a $p$-tail risk measure $\rho$, there always exists another risk measure $\rho^*$, called the generator, such that $\rho(F) = \rho^*\left(F^{[p,1]}\right)$ where $F\in\mathcal{M}$ and $F^{[p,1]}$ is the upper $p$-tail distribution of $F$. We call $(\rho,\rho^*)$ a \emph{$p$-tail pair of risk measures}.
The class of $\le_{\rm cv}$-consistent generators $\rho^*$ includes, for instance,
\begin{enumerate}[(i)]
\item
 $\rho^*=\essinf$, corresponding to $\rho = \VaR_p^R$;
 \item $\rho^*=\E$, corresponding to $\rho =\ES_p$;
  \item $\rho^*:X\mapsto -\ES_t(-X)$, corresponding $\rho=\RVaR_{p,q}$, where $t=(1-q)/(1-p)$ (see Example 5 of \cite{LW20}).
\end{enumerate}

 Introduced by \cite{WZ20}, a \emph{$p$-tail event} of a random variable $X$ is an event $A\in \mathcal{A}$ with $\mathbb{P}(A)=1-p\in (0,1)$ such that $X(\omega)\geq X(\omega')$ holds for all $\omega\in A$ and $\omega'\in A^{c}$. It is easy to check that, for  $X\sim F$, the upper $p$-tail distribution of $F$ is the same as the conditional distribution of $X$ on the $p$-tail event $A$, i.e., $F^{[p,1]}=H_{X|A}$. Therefore,  we can write the $p$-tail risk measure $\rho(X)=\rho^{*}(X_{A})$ where $X_{A}\sim H_{X|A}$ and $A$ is a $p$-tail event of $X$. Similarly, for risk aggregation $S=X+Y$, we can write the $p$-tail risk measure $\rho(S)=\rho^{*}(X_{B}+Y_{B})$ where $X_{B}\sim H_{X|B}$, $Y_{B}\sim H_{Y|B}$ and $B$ is a $p$-tail event of $S$, but not necessarily a $p$-tail event of either $X$ or $Y$.

To investigate the worst-case value of tail risk measures, we use the notion of $p$-concentration,  characterized by \cite{WZ20}. A random vector $(X,Y)$ is \emph{$p$-concentrated} if $X$ and $Y$ share a common $p$-tail event of probability $1-p$. Intuitively, for $p$ close to 1, $p$-concentrated risks will cause simultaneous large losses if the corresponding $p$-tail event happens.

There is an important connection between $p$-concentration and the worst-case risk aggregation of a $p$-tail risk measure.
In the unconstrained setting (i.e., without the order constraint), if $\rho$ is a monotone $p$-tail  risk measure, the worst-case value of $\rho$  can be attained by $p$-concentrated risks (Theorem 3 of \cite{LW20}). Earlier results of this type for VaR are Theorem 4.6 of \cite{BJW14} and Theorem 4 of \cite{bernard2017value}. Therefore, in the unconstrained setting, it suffices to look at the tail risk of each marginal distribution when we calculate the worst-case value of $\rho$. Moreover, if the generator $\rho^*$  of $\rho$ is $\le_{\rm cv}$-consistent, by \eqref{eq:RW1}, the worst-case value of $\rho$ is attained when the upper tail risks are countermonotonic.

  The following theorem studies the worst-case value of tail risk measures with the order constraint. We show that, if $(\rho,\rho^*)$ is a monotone $p$-tail pair of risk measures and $\rho^*$ is $\le_{\rm cv}$-consistent, the worst-case value of $\rho$ with the order constraint can also be attained by $p$-concentrated risks, and it is attained when the upper tail risks are DL-coupled. This result can be seen as  parallel to \citet[Theorem 3]{LW20}, which does not have the order constraint. However, the proof is quite different, and the strong stochastic order in Section \ref{sec4} through Corollary \ref{coro:1} is crucial for this result.

\begin{theorem}\label{thm2}
Suppose that $F\le_{\rm st} G$, $p\in (0,1)$, $(\rho,\rho^*)$ is a $p$-tail pair of risk measure,
and $\rho^*$ is monotone and $\le_{\rm cv}$-consistent.
We have
\begin{equation}
\label{eq:ineq2}
\overline{\rho}(\mathcal F^o_2(F,G))=\overline{\rho^*}\left(\mathcal F^o_2 \left(F^{[p,1]},G^{[p,1]} \right)\right) = \rho^*(X+Y),
\end{equation}
where $(X,Y)\sim D_*^{F^{[p,1]},G^{[p,1]}}$.
\end{theorem}

\begin{proof}
First, for any
$X\sim F^{[p,1]}$ and $Y\sim G^{[p,1]}$,
we can always construct $Z\sim F$ and $W\sim G$ such that conditional on a $p$-tail event of $Z+W$,  $ Z+W$ has the same law as $X+Y$.
This structure can be obtained by using a copula satisfying $p$-concentration.
Hence, we have the ``$\ge$" direction of  the following equality
\begin{equation}
\label{eq:ineq}
\overline{\rho}(\mathcal F^o_2(F,G))=\overline{\rho^*}(\mathcal F^o_2(F^{[p,1]},G^{[p,1]})).\end{equation}
Below, we will show  the ``$\le$" direction of \eqref{eq:ineq}.
We break the proof into several steps.
\begin{enumerate}
\item For any $X\sim F $ and $Y\sim G$ such that $X\le Y$ almost surely, let $A$ be a $p$-tail event of $X+Y$ in the sense of \cite{WZ20}.
Hence,
$\rho(X+Y) = \rho^*(X_A+Y_A)$ for some $X_A\sim H _{X|A}$ and $Y_A\sim H _{Y|A}$.
Note that here we only need to specify the distribution of $(X_A,Y_A)$, which is the conditional distribution of $(X,Y)$ on $A$.
\item
By Propositions \ref{prop:1} and \ref{prop:2}, we have $H _{Y|A}\le _{\rm st} G^{[p,1]}$.
Take $Y' \sim G^{[p,1]}$ satisfying $Y'\ge Y_A$.
The existence of $Y'$ is guaranteed by, e.g., Theorem 1.A.1 of \cite{SS07}.
  By monotonicity of $\rho^*$, we have $\rho^*(X_A+Y_A) \le \rho^*(X_A+Y')$.
\item Take $\tilde X_{A}\sim H _{X|A}$ and $\tilde Y\sim G^{[p,1]}$ such that $(\tilde X_{A},\tilde Y) $ is DL-coupled.
By $\le_{\rm cv}$-consistency  of $\rho^*$ and  Lemma \ref{lem:opt},
 we have $\rho^*(X_A +Y') \le \rho^*(\tilde X_{A}+\tilde Y)$.
\item Using Corollary \ref{coro:1}, there exists $\tilde X\sim F^{[p,1]}$
 such that
 $\tilde X_A\le \tilde X \le \tilde Y$ almost surely.
 By monotonicity of $\rho^*$, we have $\rho^*(\tilde X_A+\tilde Y) \le \rho^*(\tilde X+\tilde Y)$.

\end{enumerate}
We established the chain of inequalities
$$
\rho(X+Y) = \rho^*(X_A+Y_A) \le  \rho^*(X_A+Y')  \le \rho^*(\tilde X_A+\tilde Y) \le \rho^*(\tilde X+\tilde Y).
$$
where $\tilde X\sim F^{[p,1]}$ and $\tilde Y\sim G^{[p,1]}$.
Therefore, we obtained the ``$\le$" direction of the equality in \eqref{eq:ineq}.
The last equality in \eqref{eq:ineq2} is directly obtained from Lemma \ref{lem:opt}.
\end{proof}
\begin{remark}
For $F$, $G\in \mathcal M$, we look at cases where $\rho$ is one of VaR$^{R}$, ES and RVaR.
\begin{enumerate}[(i)]
\item
For $p\in(0,1)$, we have $\overline{\VaR}_p^R(\mathcal F^o_2(F,G))=\essinf(X+Y)$ where $(X,Y)\sim D_*^{F^{[p,1]},G^{[p,1]}}$.
\item
For $p\in (0,1)$, we have $\overline{\ES}_{p}(\mathcal F^o_2(F,G))=\E\left[F^{[p,1]}\right]+\E\left[G^{[p,1]}\right]=\ES_{p}(F)+\ES_{p}(G)$, which can also be obtained from comonotonic-additivity and subadditivity of ES. Hence the order constraint does not improve the worst-case value of ES. Indeed, the worst-case value of ES in unconstrained case is attained if and only if the two risks are $p$-concentrated (Theorem 5 of \cite{WZ20}).
\item
For $0\leq p<q<1$, we have $\overline{\RVaR}_{p,q}(\mathcal F^o_2(F,G))=-\ES_{t}(-X-Y),$ where $(X,Y)\sim D_*^{F^{[p,1]},G^{[p,1]}}$ and $t=(1-q)/(1-p)$.
\end{enumerate}
\end{remark}

Similarly, we can derive the best-case value of risk measures. For instance,
$$\underline{\RVaR}_{p,q}(\mathcal F^o_2(F,G))=-\overline{\RVaR}_{1-q,1-p}(\mathcal F^o_2(\hat G,\hat F)))=\ES_{p/q}(X+Y),$$ where $\hat G$ and $\hat F$ are the distributions of $-Y$ and $-X$, respectively, and $(X,Y)\sim D_*^{F^{[0,q]},G^{[0,q]}}$.  In Section \ref{sec5.2}, we derive analytical results for the best-case and worst-case values of VaR.

 In the following example, we calculate the worst-case value of RVaR for two uniformly distributed risks.

\begin{example}\label{ex:uniform}
For fixed $p\in(0,1)$ and distributions $F,G$ such that $F\le_{\rm st}G$, suppose that the upper $p$-tail distributions are two uniform distributions $F^{[p,1]}(x)=x$ for $x\in[0,1]$ and $G^{[p,1]}(y)=y/b$ for $y\in[0,b]$. It is easy to check that $F^{[p,1]}\le_{\rm st}G^{[p,1]}$ if and only if $b\ge 1$. We assume that $1<b<2$. Let $(X,Y)\sim D_*^{F^{[p,1]},G^{[p,1]}}$. The directional optimal transport between singular parts of $F^{[p,1]}$ and $G^{[p,1]}$ is
  $$T^{F^{[p,1]},G^{[p,1]}}(x)=\inf\left\{z \ge x:1-\frac{z}{b}<x-\frac{x}{b}\right\}=b-(b-1)x,~~~x\in[0,1].$$
Then for $c\in [0,b)$ we have $\p(X+Y\le c)=\left(\mu_{F^{[p,1]}}\wedge \mu_{G^{[p,1]}}\right)([\infty,c/2])=c/2b$. For $c\in [b,2]$,
$$\p(X+Y\le c)=(\mu_{F^{[p,1]}}\wedge \mu_{G^{[p,1]}})([\infty,c/2])+\mu_{F^{[p,1]}}'([0,(c-b)/(2-b)])=\frac{c}{2(2-b)}-\frac{b-1}{2-b}.$$
Therefore, $\VaR_{\alpha}^{R}(X+Y)=2b\alpha$ for $\alpha\in(0,1/2]$ and $\VaR_{\alpha}^{R}(X+Y)=(4-2b)\alpha+2b-2$ for $\alpha\in[1/2,1].$ By Theroem \ref{thm2}, we derive the worst-case value of RVaR
$$\overline{\RVaR}_{p,q}(\mathcal F^o_2(F,G))=
\begin{cases}
ba,~~~~q\in\left(p,\frac{1+p}{2}\right];\\
\frac{b}{4a}-\frac{1}{4a}(2a-1)(2ba-3b-4a+2),~~~~q\in\left(\frac{1+p}{2},1\right),
\end{cases}$$
where $a=1-(1-q)/(1-p)$.

\end{example}

\subsection{VaR bounds}\label{sec5.2}
The popular risk measure VaR is the most important example of a non-convex risk measure, and it is neither $\le_{\rm cx}$- nor $\le_{\rm cv}$-consistent.
In this section, we derive analytical solutions for VaR bounds with the order constraint if marginal distributions are continuous. For non-continuous marginal distributions, an algorithm is available in Section \ref{sec6} to approximate the bounds.

\begin{proposition}\label{cor:var}
For continuous distributions $F$ and $G$ such that $F\le_{\rm st} G$ and $p\in (0,1)$, we have
$$\overline{\VaR}_p^R(\mathcal F^o_2(F,G))=\min\left\{\inf_{x\in\left[F^{-1}(p+),G^{-1}(p+)\right]}\left\{T^{F^{[p,1]},G^{[p,1]}}(x)+x\right\},2G^{-1}(p+)\right\},$$

and
$$\underline{\VaR}_p^L(\mathcal F^o_2(F,G))=  \max\left\{\sup_{x\in \left[F^{-1}(p), G^{-1}(p)\right]} \left\{\hat{T}^{F^{[0,p]},G^{[0,p]}}(x)+x\right\}, 2 F^{-1}(p)\right\},$$
where $\hat T^{F^{[0,p]},G^{[0,p]}}(x)=\sup\left\{t\le x: F^{[0,p]}(t)- G^{[0,p]}(t)< F^{[0,p]}(x)-G^{[0,p]}(x)\right\}$.
\end{proposition}
\begin{proof}
For $p\in (0,1)$, as $\left(\VaR_{p}^{R},\essinf\right)$ is a $p$-tail pair of risk measures and $\essinf$ is $\le_{\rm cv}$-consistent, by Theorem \ref{thm2},
\begin{equation*}
\overline{\VaR}^{R}_p(\mathcal F^o_2(F,G))=\overline{\essinf}(\mathcal F^o_2(F^{[p,1]},G^{[p,1]})) = \essinf(X+Y),
\end{equation*}
where $(X,Y)\sim D_*^{F^{[p,1]},G^{[p,1]}}$. By Theorem \ref{essinf}, we obtain the first result.
For the second result, let $X' \sim F$ and $Y'\sim G$.  Denote by $\hat F$ and $\hat G$ the distributions of $-X'$ and $-Y'$. We have
$$\underline{\VaR}_p^L(\mathcal F^o_2(F,G))=-\overline {\VaR}_{1-p}^R(\mathcal F^o_2(\hat G,\hat F))=\underline \esssup\left( \mathcal F^o_2\left(F^{[0,p]},G^{[0,p]}\right)\right).$$
Applying Theorem \ref{essinf}, we get the desired result.
\end{proof}
\begin{remark}
For $F,G\in \mathcal{M}$ and $p\in(0,1)$, the worst-case value of VaR$^{R}_{p}$ and the best-case value of VaR$^{L}_{p}$ without the order constraint are attained by letting the upper tail risks and lower tail risks be countermonotonic, respectively, i.e.,
$$\sup\left\{\VaR^{R}_{p}(X+Y): X\sim F,~Y\sim G\right\}=\inf_{x\in[0,1-p]}\left\{F^{-1}(p+x)+G^{-1}(1-x)\right\},$$
$$\inf\left\{\VaR^{L}_{p}(X+Y): X\sim F,~Y\sim G\right\}= \sup_{x\in[0,p]}\left\{F^{-1}(x)+G^{-1}(p-x)\right\}. $$
See \cite{M81} and \cite{R82}.
\end{remark}

\begin{example}[Pareto risks: VaR bounds]\label{ex:varbounds}
Following the marginal assumptions on $F$ and $G$ in Example \ref{ex:pareto}, we derive $\overline{\VaR}^{R}_{p}(\mathcal F^o_2(F,G))$ and $\underline{\VaR}_p^L(\mathcal F_2^o(F,G))$ by Proposition \ref{cor:var}. For $p \in (0,1)$, we have
$$F^{[p,1]}(x)=\(1-\frac{1}{x(1-p)}\)\id_{\{x\ge 1/(1-p)\}}~~~~ \text{and} ~~~~ G^{[p,1]}(x)=\(1-\frac{2}{x(1-p)}\)\id_{\{x\ge 2/(1-p)\}}.$$
Thus,
$$T^{F^{[p,1]},G^{[p,1]}}(x)=\frac{x}{x(1-p)-1}, ~~~~ x\in \left(\frac{1}{1-p},\frac{2}{1-p}\right].$$
And $T^{F^{[p,1]},G^{[p,1]}}\left(1/(1-p)\right)=\inf\{\emptyset\}=\infty$.  Therefore,
$$\overline{\VaR}^{R}_{p}(\mathcal F^o_2(F,G))=\min\left\{\inf_{x\in [1/(1-p),2/(1-p)]}\left\{\frac{x}{x(1-p)-1}+x\right\},\frac{4}{(1-p)}\right\}=\frac{4}{1-p}.$$
Similarily, we have $\underline{\VaR}_{p}^{L}(\mathcal F^o_2(F,G))=
1+2/(1-p)$. Those bounds on VaR will be used to calculate probability bounds of $X+Y$ in Example \ref{ex:prob}, where $X\sim F$ and $Y\sim G$.
\end{example}

In the unconstrained problem, \cite{BJW14} showed that the worst-case value of  VaR$^{R}_{p}$ is a continuous function of $p\in(0,1)$ if the marginal distributions are strictly increasing. This continuity result is used to confirm that there is no need to distinguish between VaR$^{R}$ and VaR$^{L}$ when we calculate their worst-case values (best-case values). We will see later that the above statement is still true if the order constraint is further imposed.
The continuity of the worst-case value of VaR$^{R}$ with  the order constraint is established in
  Lemma \ref{lem:conto}. The proof of Lemma \ref{lem:conto} is surprisingly complicated, very different from the case treated by \cite{BJW14}, and it is put in the Appendix.

\begin{lemma}\label{lem:conto}
  For strictly increasing continuous  distribution functions $F$ and $G$ such that $F\le_{\rm{st}}G$, the function  $p \mapsto \overline{\VaR}_p^R(\mathcal F^o_2(F,G))$ is  continuous  on  $ (0,1)$.
\end{lemma}

Using Lemma \ref{lem:conto}, we obtain that the worst-case values (best-case values) of VaR$^{R}$ and VaR$^{L}$ with the order constraint are equivalent for strictly increasing continuous distributions.
\begin{proposition}\label{VaR:cont}
Suppose that $F$ and $G$ are strictly increasing continuous distribution functions such that $F \le_{\rm{st}}G$. For $p\in (0,1)$, we have
$$\overline{\VaR}_p^L(\mathcal F^o_2(F,G))=\overline{\VaR}_p^R(\mathcal F^o_2(F,G))~~~~\text{and}~~~~\underline{\VaR}_p^L(\mathcal F^o_2(F,G))=\underline{\VaR}_p^R(\mathcal F^o_2(F,G)).$$
\end{proposition}
\begin{proof}
For $\epsilon>0$, we have $\overline{\VaR}^{R}_{p-\epsilon}(\mathcal F^o_2(F,G))\le \overline{\VaR}_p^{L}(\mathcal F^o_2(F,G))\le \overline{\VaR}^{R}_p(\mathcal F^o_2(F,G))$. By Lemma \ref{lem:conto}, $\overline{\VaR}_p^R(\mathcal F^o_2(F,G))$ is a continuous function of $p \in (0,1)$. Letting $\epsilon \downarrow 0$, we get the desired result for worst-case value of VaR$^{L}$ and VaR$^{R}$. The proof for the best-case value of VaR$^{L}$ and VaR$^{R}$ is similar and thus omitted.
\end{proof}
By Proposition \ref{VaR:cont}, in  practical situations of risk management,
 there is no need to distinguish between VaR$^{L}_{p}$ and VaR$^{R}_{p}$ when we calculate their bounds with  the order constraint; this observation will be useful in the numerical studies in Section \ref{sec6}.

\subsection{Probability bounds}\label{prob}
In risk management and quantitative finance,  probability bounds of the aggregate position are also of great interest. In the unconstrained problem (i.e., only marginal distributions are known),
the probability bounds
on the aggregation of two risks are given by \cite{R82}.

 For $F,G\in \mathcal{M}$ such that $F\le_{\rm{st}} G$ and $t\in\R$,  we are interested in the upper and lower bounds of probability with the order constraint, defined as
$$M^{o}(t):=\sup\left\{\p(X+Y\le t):(X,Y)\in \mathcal{F}_{2}^{o}(F,G)\right\}
$$
and
$$m^{o}(t):=\inf \left\{\p(X+Y< t):(X,Y)\in \mathcal{F}_{2}^o(F,G)\right\}.$$
The above upper and lower bounds of probability can be obtained by inverting the lower and upper bounds of VaR, respectively.
In particular, for $p\in(0,1)$, we have
$$\overline{\VaR}^{R}_p(\mathcal F_2^o(F,G))=(m^{o})^{-1}(p)\text{~~~and~~~}\underline{\VaR}^{L}_p(\mathcal F_2^o(F,G))=(M^{o})^{-1}(p).$$
For continuous marginal distributions, we can invert the analytical solutions in Proposition \ref{cor:var} to obtain the probability bounds with the order constraint. While the analytical solutions to VaR bounds does not necessarily lead to an explicit results for probability bounds, a numerical algorithm in Section \ref{sec6} can be used to approximate probability bounds. The following example compares probabilities bounds with and without order constraint for Pareto marginal distributions.
\begin{example}[Pareto: Probability bounds] \label{ex:prob}
Following the   assumptions in Examples \ref{ex:pareto} and \ref{ex:varbounds}, we convert the VaR bounds in Example \ref{ex:varbounds} to obtain   probability bounds with the order constraint:
$$M^{o}(t)=1-\frac{4}{t},~~~t\ge 4, \text{~~~and~~~}m^{o}(t)=1-\frac{2}{t-1},~~~t\ge 3.$$
The probability bounds with and without order constraint are plotted in Figure \ref{fig:prob}. The bounds without  the order constraint are denoted by $M$ and $m$. The figure shows that the order constraint improves the lower probability bound a lot while there is no improvement for the upper bound (the difference between $M^o$ and $M$ is invisible). When two risks are countermonotonic or DL-coupled, the corresponding probability   (denoted by $\mathrm{Prob}^{\rm CT}$ and $\mathrm{Prob}^{\rm DL}$, respectively) lies between the constrained bounds for $t\ge 8$.
\begin{figure}[htbp]
\centering
\caption{Probability bounds in Example \ref{ex:prob}}\label{probbound}
\label{fig:prob}
\includegraphics[height=7cm]{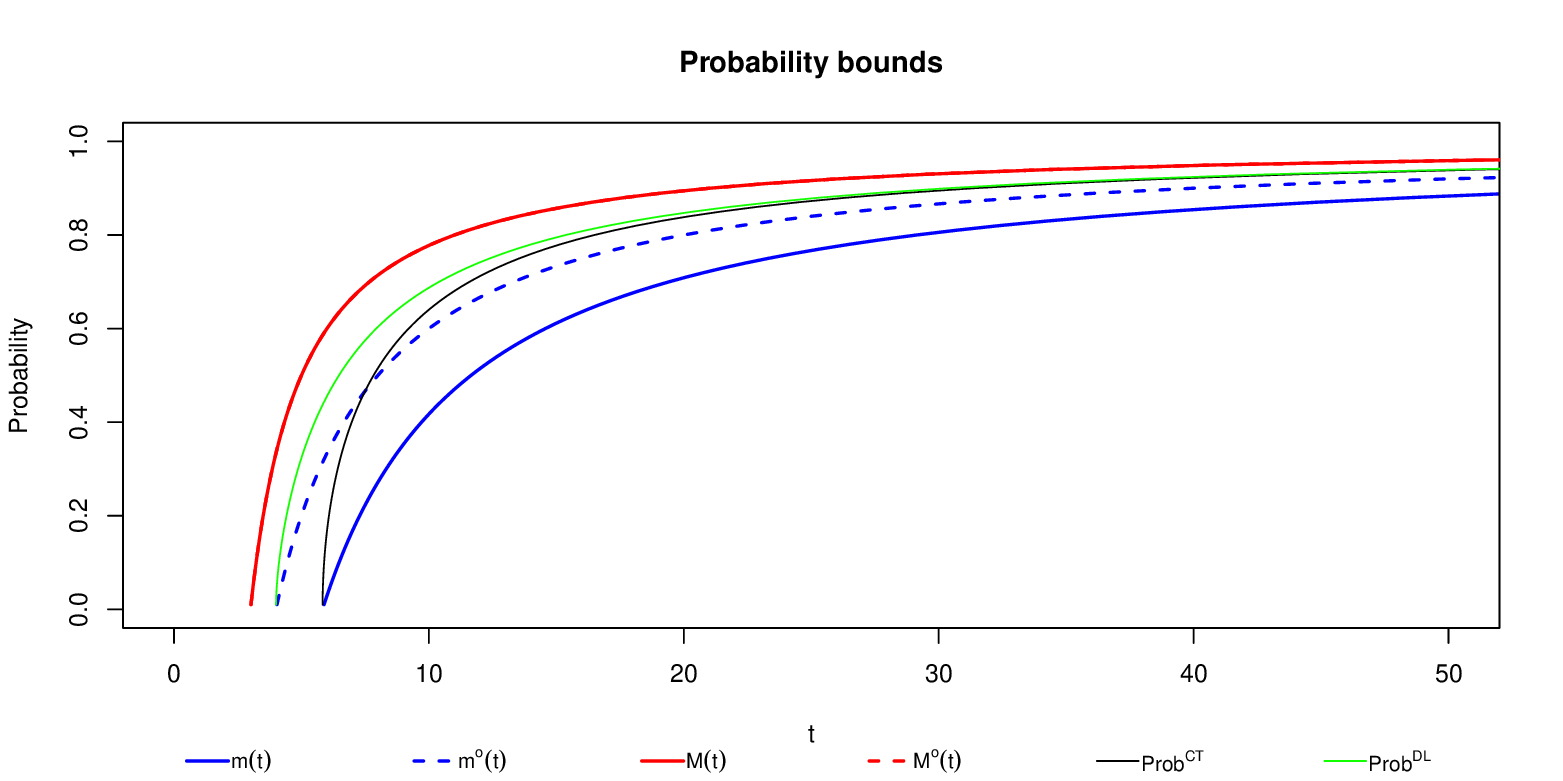}
\end{figure}
\end{example}

\section{Numerical results and a real-data application}\label{sec6}

In this section, we use numerical examples and a case study to illustrate the impact of the order constraint on VaR bounds (the worst-case and best-case values of $\VaR^{R}$ and $\VaR^{L}$), and RVaR bounds. For convenience, we do not distinguish between $\VaR^{L}$ and $\VaR^{R}$ when we calculate their bounds (Proposition \ref{VaR:cont}). Both $\VaR^{L}$ and $\VaR^{R}$ are referred to  as $\VaR$ in numerical results. We only illustrate the numerical calculations for VaR bounds. RVaR bounds can be calculated in a similar manner.

\subsection{General methodology}
Let  $F,G\in \mathcal {M}$ be continuous distributions such that $F\le_{\rm st}G$. As suggested by Theorem \ref{thm2}, the best-case  and worst-case values of VaR are determined by the lower tail and upper tail distributions of $F$ and $G$, respectively.
To approximately calculate $\overline{\VaR}^{R}_{p}$ for $p\in (0,1)$, we first discretize the upper $p$-tail distributions $F^{[p,1]}$ and $G^{[p,1]}$. Fix an integer $n$ and let
$$x_{i}=F^{-1}\left(p+\frac{(1-p)(n-i)}{n}\right) \text{   and   } y_{i}=G^{-1}\left(p+\frac{(1-p)(n-i)}{n}\right),$$
for $i=1,\dots,n$. Define $S_{X}^{[p,1]}=\{x_{1},\dots,x_{n}\}$ and $S_{Y}^{[p,1]}=\{y_{1},\dots,y_{n}\}$. If $S_{X}^{[p,1]}$ and $S_{Y}^{[p,1]}$ have no identical locations, we use the following algorithm introduced in \cite{NW20} to approximate the DL coupling between $F^{[p,1]}$ and $G^{[p,1]}$. Let $S_{1}=S_{Y}^{[p,1]}$, we iterate for $k=1,\dots,n$
\begin{enumerate}[(i)]
\item
$T(x_{k}):=\min\left\{y\in S_{k}:y\geq x_{k}\right\}$,
\item
$S_{k+1}=S_{k}\backslash\left\{T(x_{k})\right\}$.
\end{enumerate}
Let $s_{k}=x_{k}+T(x_{k})$, $k=1,\dots,n$. We use $\min\{s_{k}:k=1,\dots,n\}$
 as the approximation for $\overline{\VaR}^{R}_{p}$.
 The best-case value of VaR$^{L}$ can be obtained in a similar manner.  The unconstrained bounds of VaR, attained by conditional countermonotonicity, can be numerically computed by the Rearrangement Algorithm (RA) in \cite{PR12} and \cite{EPR13}.
 Similar procedures can be constructed for discrete distributions.

The difference between the worst-case and best-case values of a risk measure is called the \emph{Dependence Uncertainty spread} (DU-spread) for the risk measure, which is used as a measure of dependence uncertainty (see \cite{EWW15}). We use the DU-spread reduction defined in \cite{puccetti2017reduction} to measure the improvement on unconstrained VaR bounds due to the order constraint. Denote by $L$ and $U$ the unconstrained best-case and worst-case values of a risk measure $\rho$. Similarly, denote by $L^{o}$ and $U^{o}$ the bounds with the order constraint.
 The lower and upper reductions of DU-spread are defined as
\begin{align}
\label{eq:rw-april3}
R^{L}=\frac{L^{o}-L}{U-L}~~\text{and}~~~R^{U}=\frac{U-U^{o}}{U-L}.
\end{align}
The DU-spread reduction is defined as the sum of lower and upper DU-spread reductions, which is $R=R^{L}+R^{U} \in [0,1]$.

\subsection{Numerical examples}

Consider distributions $F$ and $G_{i}$  such that their means are 50 and $50+10i$ and $F\le _{\rm st}G_{i}$, $i=1,2,3$. The distributions are specified in uniform and Pareto cases as below.
\begin{table}[htbp]
\centering
\caption{Distributions for numerical examples}
\label{t3}{\small
\begin{tabular}{c|l|l|l|l}
  \hline
    Uniform & $F(x)=x/100$ & $G_{1}(x)=x/120$& $G_{2}(x)=x/140$ & $G_{3}(x)=x/160$\\
    \hline
    Pareto & $F(x)=1-(25/x)^{2}$ & $G_{1}(x)=1-(30/x)^{2}$& $G_{2}(x)=1-(35/x)^{2}$ & $G_{3}(x)=1-(40/x)^{2}$\\
  \hline

\end{tabular}
}
\end{table}

For both $F$ and $G_{i}$,  $i=1,2,3$, being uniform or Pareto distributions, we calculate the improvement (i.e., reduction of DU-spread) on VaR bounds and RVaR bounds. We also present the results of VaR$_{p}(X+Y)$ if $X\sim F$ and $Y\sim G_{i}$ are independent, comonotonic, DL-coupled and countermonotonic, $i=1,2,3$.  The results of VaR bounds for uniform and Pareto cases can be found in Figures \ref{ubound} and \ref{pbound}, respectively. The results  of RVaR bounds can be found in Tables \ref{ubound:RVaR} and \ref{pbound:RVaR}. We make the following observations.
\begin{enumerate}[(i)]
\item
 The DU-spread reductions in all tables and figures show that the improvement due to the order constraint is significant for both VaR and RVaR. The improvement for VaR becomes larger as $p$ increases from 0.9 to 1.
\item
For all uniform and Pareto cases, as the mean of $G_{i}$ becomes larger, the improvement becomes smaller. In other words, the more ``similar'' the distributions $F$ and $G_{i}$ are, the more improvement is gained from imposing the order constraint. 
\item
 The order constraint has an overall larger improvement on the bounds for uniform distributions than those for Pareto distributions. Nevertheless, the improvement on the worst-case value is insignificant for uniform distributions. This is because $\essinf G_{i}^{[p,1]}\ge \esssup F^{[p,1]}$ for $p\in(0.9,1)$, and the DL coupling of the upper $p$-tail distributions is the same as countermonotonicity. While for Pareto distributions, the improvement on the worst-case value is even larger than that on the best-case value.
 \item
 For the uniform cases, if the risks are countermonotonic,  both VaR and RVaR are close to the unconstrained lower bound.  If the risks are DL-coupled,  both VaR and RVaR are close to the constrained lower bound for the uniform cases while they lie between the constrained bounds for the Pareto cases. If the risks are comonotonic,  both VaR and RVaR lie between the constrained bounds for all cases.

\end{enumerate}

\begin{figure}[htbp]
\centering
\caption{Uniform cases: $\VaR_{p}$ bounds, DU reduction and $\VaR_{p}$ of the aggregate risk with different dependence structures are contained in this figure. VaR values with independence, comonotonicity, countermonotonicity and DL coupling are denoted by $\VaR^{\rm Ind}$, $\VaR^{\rm C}$, $\VaR^{\rm Co}$ and $\VaR^{\rm DL}$, respectively.  }\label{ubound}
\includegraphics[height=5cm]{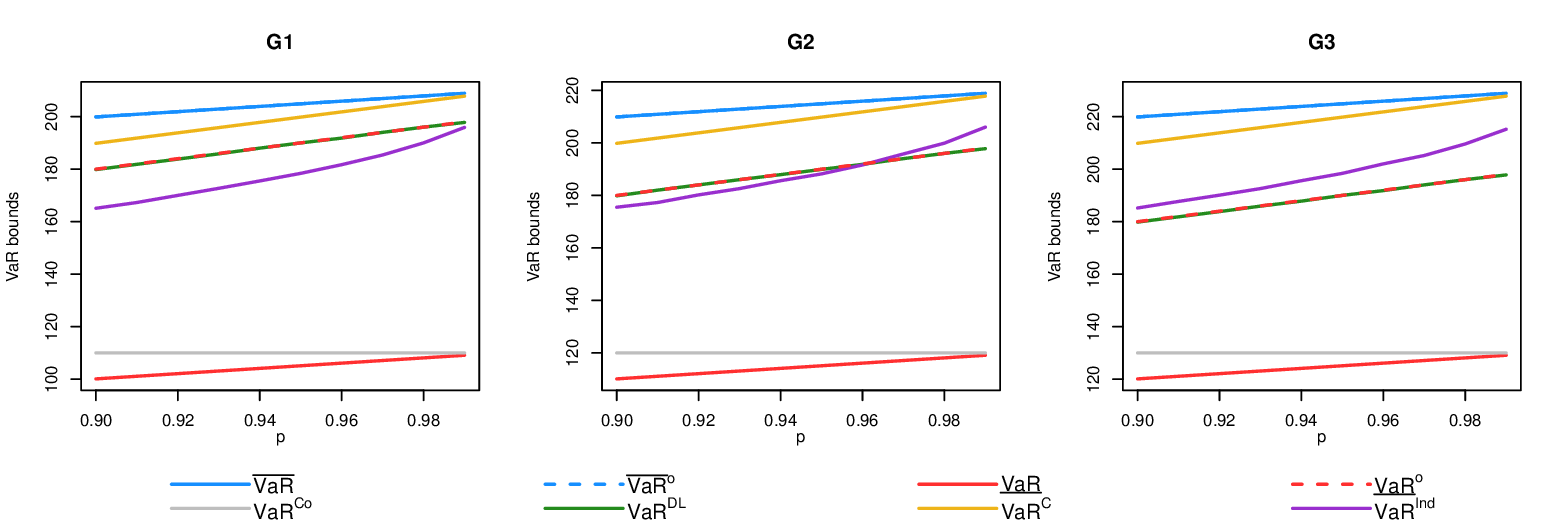}
\includegraphics[height=5cm]{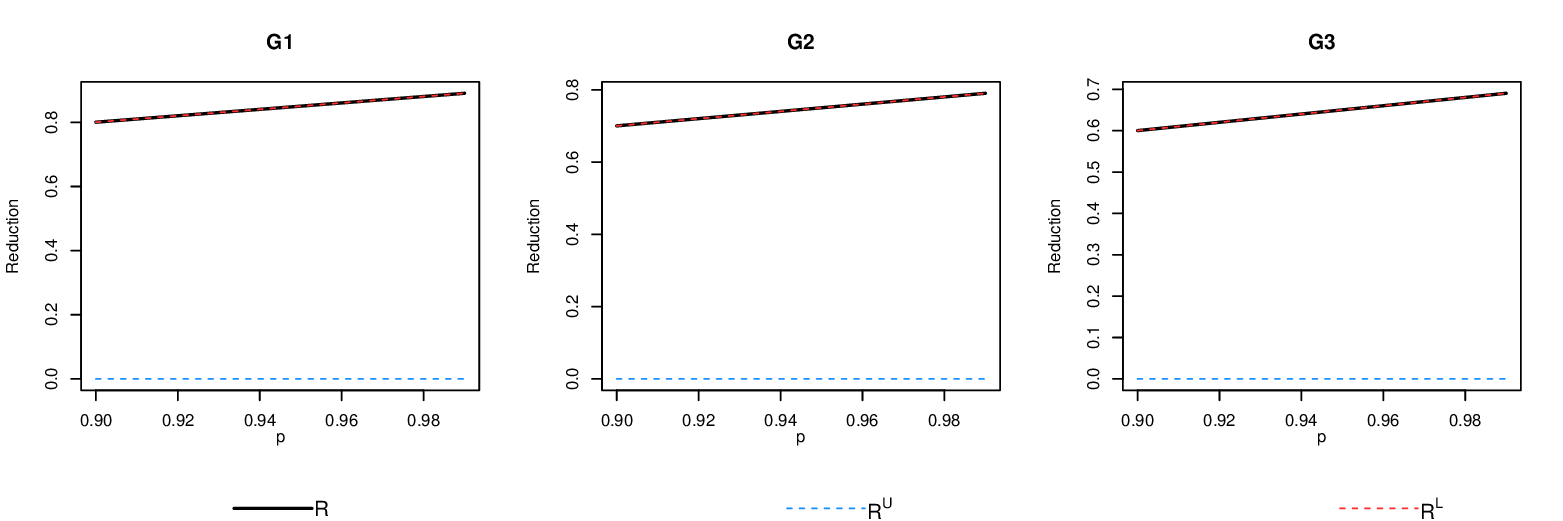}
\end{figure}

\begin{figure}[htbp]
\centering
\caption{Pareto cases: $\VaR_{p}$ bounds, DU reduction and $\VaR_{p}$ of the aggregate risk with different dependence structures are contained in this figure. VaR values with independence, comonotonicity, countermonotonicity and DL coupling are denoted by $\VaR^{\rm Ind}$, $\VaR^{\rm C}$, $\VaR^{\rm Co}$ and $\VaR^{\rm DL}$, respectively.  }\label{pbound}
\includegraphics[height=5cm]{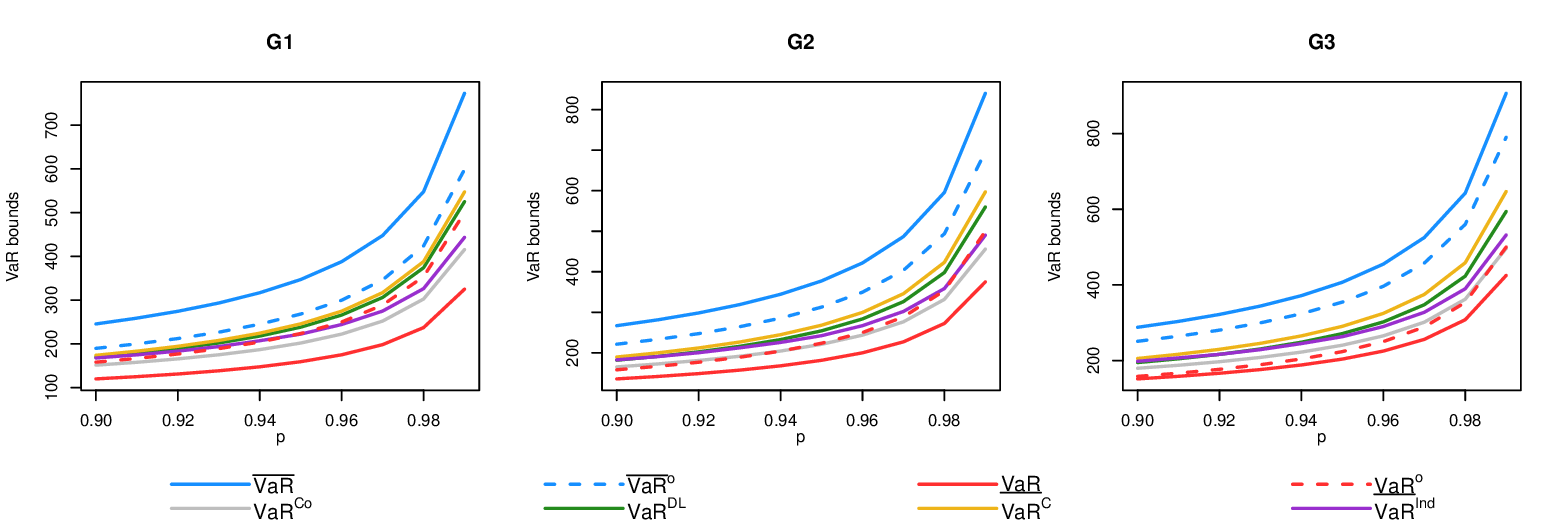}
\includegraphics[height=5cm]{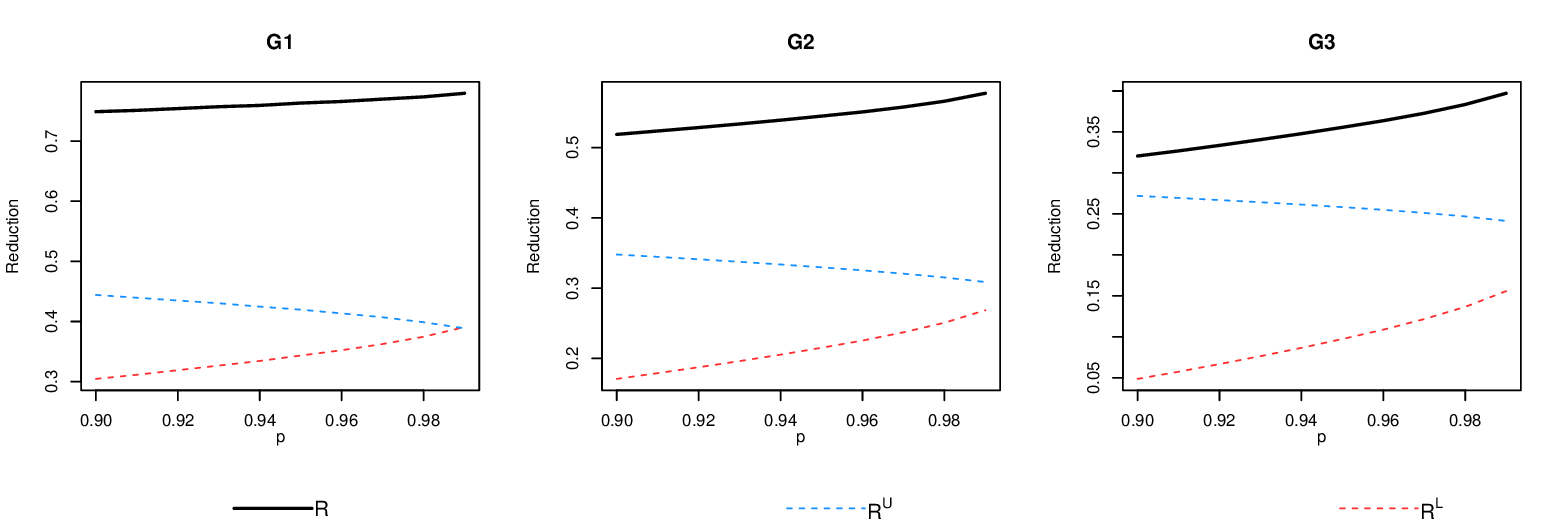}
\end{figure}

\begin{table}[htbp]
\centering
\caption{Uniform cases: $\RVaR_{p,q}$ bounds, DU reduction and $\RVaR_{p,q}$ of the aggregate risk with different dependence structures are contained in this table. Marginal distributions in Case $i$ are uniform distributions $F$ and $G_i$ given in Table \ref{t3}. }
\label{ubound:RVaR}
{\scriptsize
\begin{tabular}{cccccc}
    \hline
    &$p$ & $75\%$& $90\%$ & $95\%$  & $99.5\%$ \\
    &$q$ & $90\%$& $95\%$ & $99.5\%$  & $99.9\%$ \\
    \hline
   \multirow{7}{*}{Case 1}&  Constrained bounds & $( 165 , 182 )$ & $( 185 , 200 )$ & $( 195 , 205 )$ & $( 200 , 209 )$ \\
  &Unconstrained bounds & $( 100 , 185 )$ & $( 105 , 200 )$ & $( 110 , 205 )$ & $( 110 , 209 )$ \\
  &$(R^{L},R^{U},R)$ & $( 0.77 , 0.04 , 0.81 )$ & $( 0.84 , 0 , 0.84 )$ & $( 0.89 , 0 , 0.89 )$ & $( 0.9 , 0 , 0.9 )$ \\ [3mm]
  &Independence & 151 & 171 & 187 & 203 \\
  &Comonotonicity & 175 & 195 & 204 & 209 \\
  &Countermonotonicity & 110 & 110 & 110 & 110 \\
  &DL coupling  & 165 & 185 & 194 & 199 \\

   \hline

   \multirow{7}{*}{Case 2}&  Constrained bounds & $( 165 , 195 )$ & $( 185 , 210 )$ & $( 195 , 215 )$ & $( 200 , 219 )$ \\
  &Unconstrained bounds & $( 110 , 195 )$ & $( 115 , 210 )$ & $( 120 , 215 )$ & $( 120 , 219 )$ \\
  &$(R^{L},R^{U},R)$ & $( 0.65 , 0 , 0.65 )$ & $( 0.74 , 0 , 0.74 )$ & $( 0.79 , 0 , 0.79 )$ & $( 0.8 , 0 , 0.8 )$ \\ [3mm]
  &Independence & 161 & 181 & 197 & 212 \\
  &Comonotonicity & 185 & 205 & 214 & 219 \\
  &Countermonotonicity & 120 & 120 & 120 & 120 \\
  &DL coupling  & 165 & 185 & 195 & 199 \\

   \hline

   \multirow{7}{*}{Case 3}& Constrained bounds & $( 165 , 205 )$ & $( 185 , 220 )$ & $( 195 , 225 )$ & $( 200 , 229 )$ \\
  &Unconstrained bounds & $( 120 , 205 )$ & $( 125 , 220 )$ & $( 130 , 225 )$ & $( 130 , 229 )$ \\
  &$(R^{L},R^{U},R)$ & $( 0.53 , 0 , 0.53 )$ & $( 0.63 , 0 , 0.63 )$ & $( 0.68 , 0 , 0.68 )$ & $( 0.7 , 0 , 0.7 )$ \\ [3mm]
  &Independence & 171 & 192 & 207 & 222 \\
  &Comonotonicity & 195 & 215 & 224 & 229 \\
  &Countermonotonicity & 130 & 130 & 130 & 130 \\
  &DL coupling  & 165 & 185 & 195 & 199 \\

   \hline

\end{tabular}
}
\end{table}

\begin{table}[htbp]
\centering
\caption{Pareto cases: $\RVaR_{p,q}$ bounds, DU reduction and $\RVaR_{p,q}$ of the aggregate risk with different dependence structures are contained in this table. Marginal distributions in Case $i$ are Pareto distributions $F$ and $G_i$ given in Table \ref{t3}.}
\label{pbound:RVaR}
{\scriptsize
\begin{tabular}{cccccc}
    \hline
    &$p$ & $75\%$& $90\%$ & $95\%$  & $99.5\%$ \\
    &$q$ & $90\%$& $95\%$ & $99.5\%$  & $99.9\%$ \\
    \hline
   \multirow{7}{*}{Case 1}& Constrained bounds & $( 125 , 140 )$ & $( 185 , 213 )$ & $( 354 , 379 )$ & $( 1012 , 1085 )$ \\
  &Unconstrained bounds & $( 103 , 164 )$ & $( 140 , 254 )$ & $( 262 , 409 )$ & $( 679 , 1209 )$ \\
  &$(R^{L},R^{U},R)$ & $( 0.35 , 0.39 , 0.75 )$ & $( 0.4 , 0.36 , 0.76 )$ & $( 0.63 , 0.21 , 0.83 )$ & $( 0.63 , 0.23 , 0.86 ) $\\ [3mm]
  &Independence & 136 & 191 & 316 & 800 \\
  &Comonotonicity & 135 & 204 & 373 & 1063 \\
  &Countermonotonicity & 124 & 172 & 292 & 774 \\
  &DL coupling  & 132 & 198 & 360 & 1017 \\

   \hline

   \multirow{7}{*}{Case 2}& Constrained bounds & $( 129 , 157 )$ & $( 188 , 240 ) $& $( 371 , 418 )$ & $( 1042 , 1204 )$ \\
  &Unconstrained bounds & $( 114 , 178 )$ & $( 156 , 276 )$ & $( 289 , 446 )$ & $( 755 , 1316 )$ \\
  &$(R^{L},R^{U},R)$ & $( 0.24 , 0.33 , 0.57 )$ & $( 0.27 , 0.3 , 0.57 )$ & $( 0.53 , 0.18 , 0.7 )$ & $( 0.51 , 0.2 , 0.71 )$ \\ [3mm]
  &Independence & 148 & 208 & 345 & 885 \\
  &Comonotonicity & 147 & 222 & 407 & 1160 \\
  &Countermonotonicity & 135 & 188 & 320 & 851 \\
  &DL coupling  & 143 & 212 & 383 & 1075 \\

   \hline

   \multirow{7}{*}{Case 3}& Constrained bounds & $( 135 , 173 )$ & $( 194 , 265 )$ & $( 391 , 456 )$ & $( 1083 , 1321 )$ \\
  &Unconstrained bounds & $( 125 , 192 )$ & $( 174 , 298 )$ & $( 317 , 482 )$ & $( 838 , 1421 )$ \\
  &$(R^{L},R^{U},R)$ & $( 0.15 , 0.29 , 0.44 )$ & $( 0.17 , 0.27 , 0.43 )$ & $( 0.45 , 0.16 , 0.6 )$ & $( 0.42 , 0.17 , 0.59 )$ \\ [3mm]
  &Independence & 161 & 225 & 375 & 972 \\
  &Comonotonicity & 159 & 241 & 441 & 1256 \\
  &Countermonotonicity & 147 & 205 & 349 & 932 \\
  &DL coupling  & 153 & 226 & 407 & 1144 \\

   \hline

\end{tabular}
}
\end{table}

As the observations on VaR and RVaR are similar, we will focus on studying VaR for the rest of the paper. In previous examples for Pareto distributions, the tail parameter\footnote{We use the Pareto$(\theta,\alpha)$ distribution parametrized by $F(x)=1-\(\theta/x\)^{\alpha}$ for $x\ge \theta$, where $\theta\in \R$ is the location parameter and $\alpha >0$ is the tail parameter.} of distributions $F$ and $G$ are fixed (see Table \ref{t3}). Next, we study the improvement of VaR bounds as the tail parameter of the distribution $G$ varies. Let $F(x)=1-(25/x)^2$ and $G(x)=1-(25/x)^{\alpha}$, $\alpha\le 2$. VaR bounds are calculated as $\alpha$ increases from 1.3 to 2. Results can be found in Figure \ref{alpha}. We observe that the larger $\alpha$ is, the greater the improvement is gained from the order constraint. However, the improvement on the unconstrained lower bound is neligible for small $\alpha$. As in previous examples for Pareto distributions, comonotonicity and DL coupling produce very close VaR values.
\begin{figure}[t]
\centering
\caption{This figure contains $\VaR_{p}$ bounds, DU reduction and $\VaR_{p}$ of the aggregated Pareto risks with different dependence structures as $\alpha$ changes, where $F=\mathrm{Pareto}(25,2)$ and $G=\mathrm{Pareto}(25,\alpha)$. VaR values with independence, comonotonicity, countermonotonicity and DL coupling are denoted by $\VaR^{\rm Ind}$, $\VaR^{\rm C}$, $\VaR^{\rm Co}$ and $\VaR^{\rm DL}$, respectively.}\label{alpha}
\includegraphics[height=5cm]{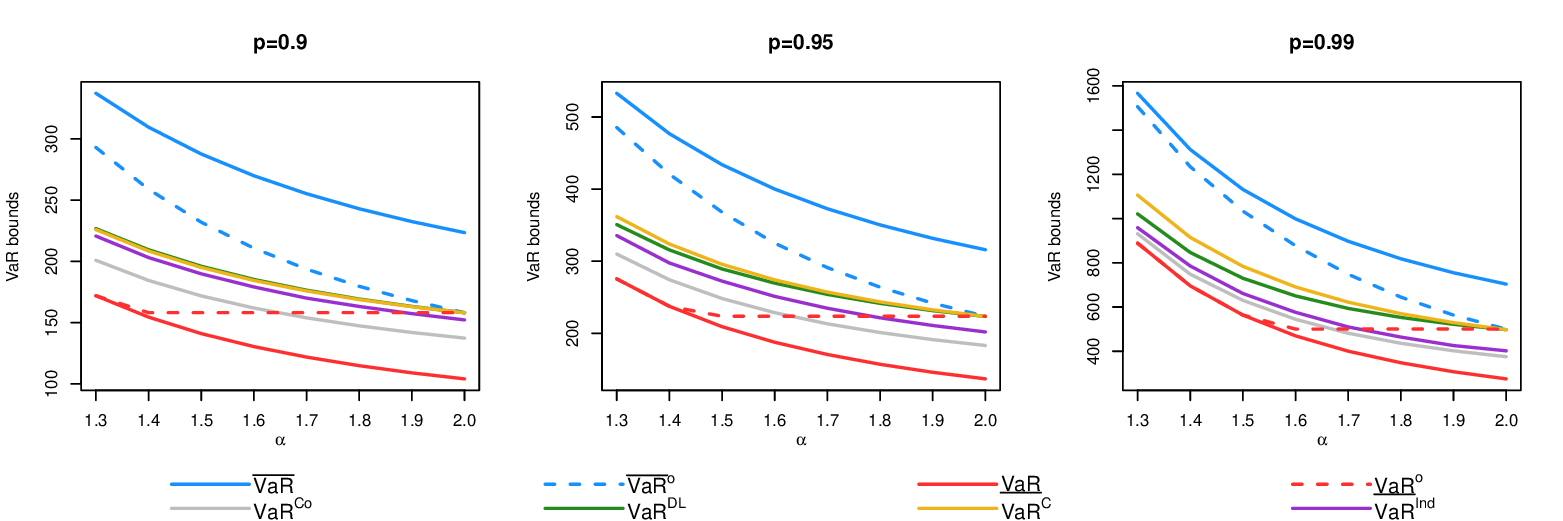}
\includegraphics[height=5cm]{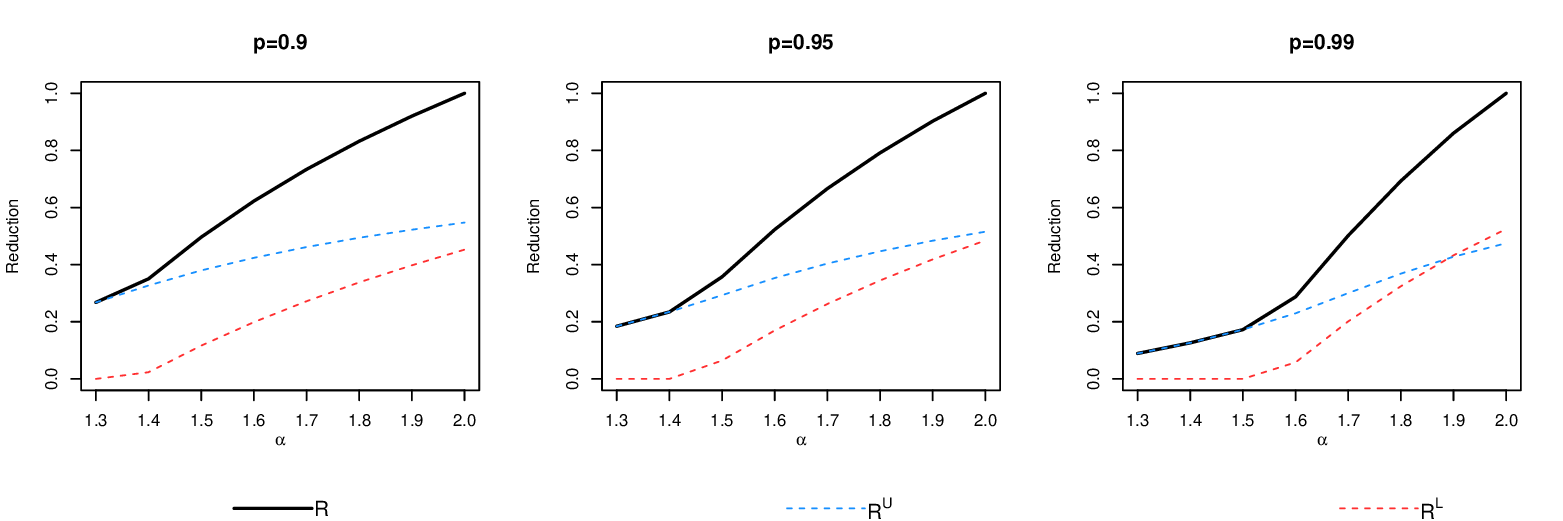}
\end{figure}

\subsection{Case study: Health insurance policies}\label{sec6.2}

In this case study, we calculate the bounds of VaR with and without order constraint for a health insurance portfolio.
Insurance policies can be classified according to certain characteristics of policyholders. For illustration purposes,  we use gender to make classifications on health insurance policies (this may not be allowed in certain countries).
The aggregate loss of the portfolio can be expressed as $S=X+Y$ where $X\sim F$ and $Y\sim G$ represent the losses caused by females and males, respectively, from a portfolio of $50$ males and $50$ females.
It is sensible to guess that $F \le_{\rm st} G$, due to the morbidity differences between males and females; this will be confirmed by our dataset.
 Moreover, since the losses by males and females are affected by many common factors, the assumption that $X\le Y$ seems also reasonable.

 We use the Hospital Costs data of \cite{frees2009regression} which were originally from the Nationwide Inpatient Sample of the Healthcare Cost and Utilization Project (NIS-HCUP), to represent the individual losses of the health insurance policies. The data contains 500 observations with 244 males and 256 females. We generate $1000$ bootstrapping samples of the total losses caused by $50$ males and $50$ females, respectively.

  The empirical distributions  $\hat F$ and $\hat G$ from the 1000 bootstrapping samples are plotted in the top-left panel of Figure \ref{cdf}.
Although $\hat F$ and $\hat G$ do not satisfy $\hat F \le_{\rm st} \hat G$,  such a violation is almost invisible (see the bottom-left panel of Figure \ref{cdf}) and possibly caused by  sampling randomness.
 Indeed,   using the Kolmogorov-Smirnov-type test of \cite{barrett2003consistent}, we cannot reject the hypothesis  $F \le_{\rm st} G$ for the bootstrap data. 
  Hence, $F \le_{\rm st} G$ is   a sensible assumption.
  The  isotonic distributional regression (IDR), introduced by \cite{henzi2019isotonic}, is a nonparametric technique to estimate distributions with order restrictions (e.g., stochastic order and hazard rate order).  We use IDR to estimate $F$ and $G$ such that the stochastic order holds for the estimated distributions. The estimated distributions are plotted in the top-right panel of Figure \ref{cdf}, and they are used to calculate the VaR bounds.
  However, if $\hat F \le_{\rm st} \hat G$ holds already, IDR is not necessary, and we can directly use the empirical distributions.

\begin{figure}[t]
\centering
\caption{Empirical and estimated distributions of $X$ and $Y$. Top panels: entire region; bottom panels: tail region}\label{cdf}
\includegraphics[height=5.5cm]{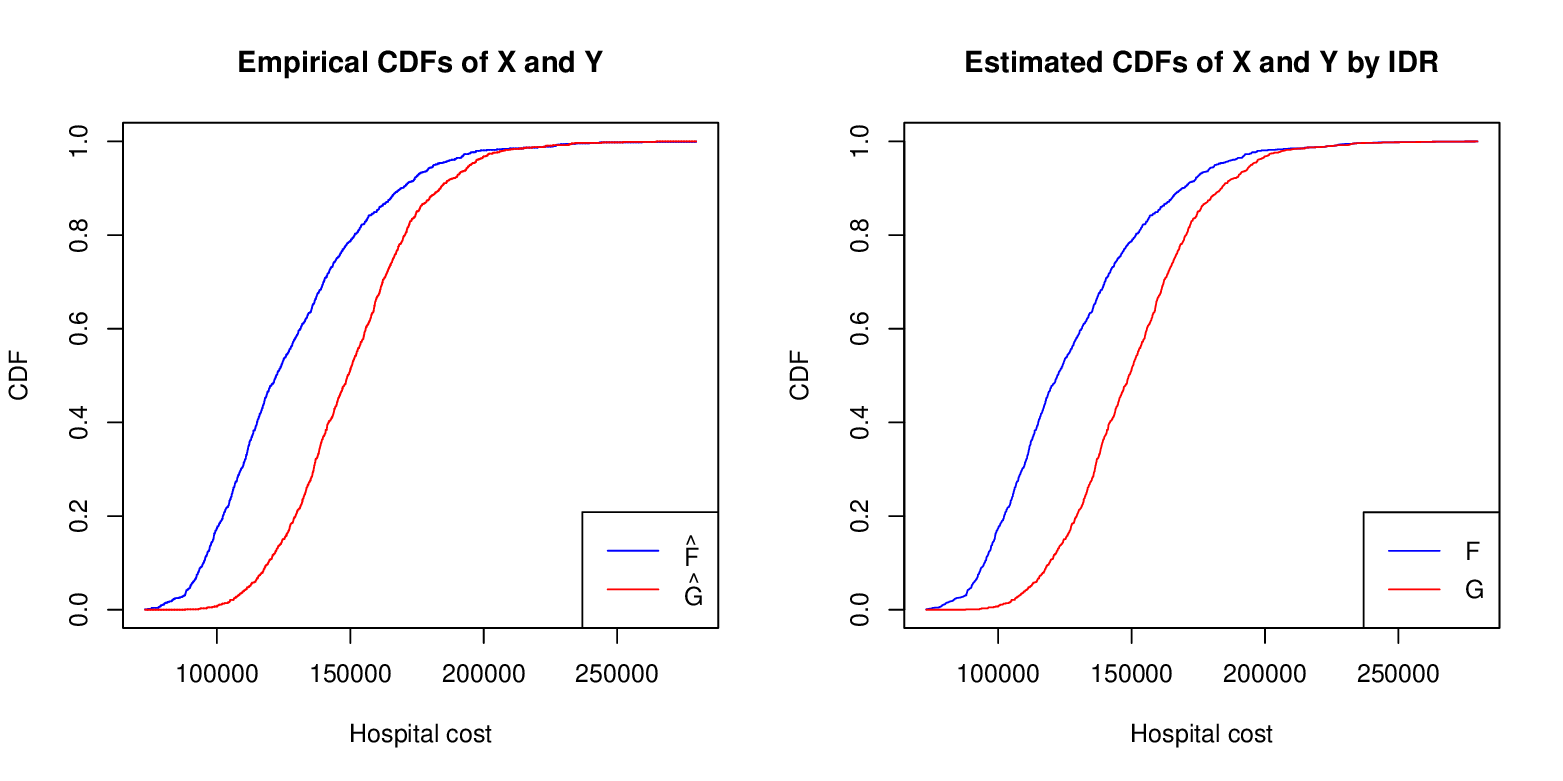}\\
\includegraphics[height=5.5cm]{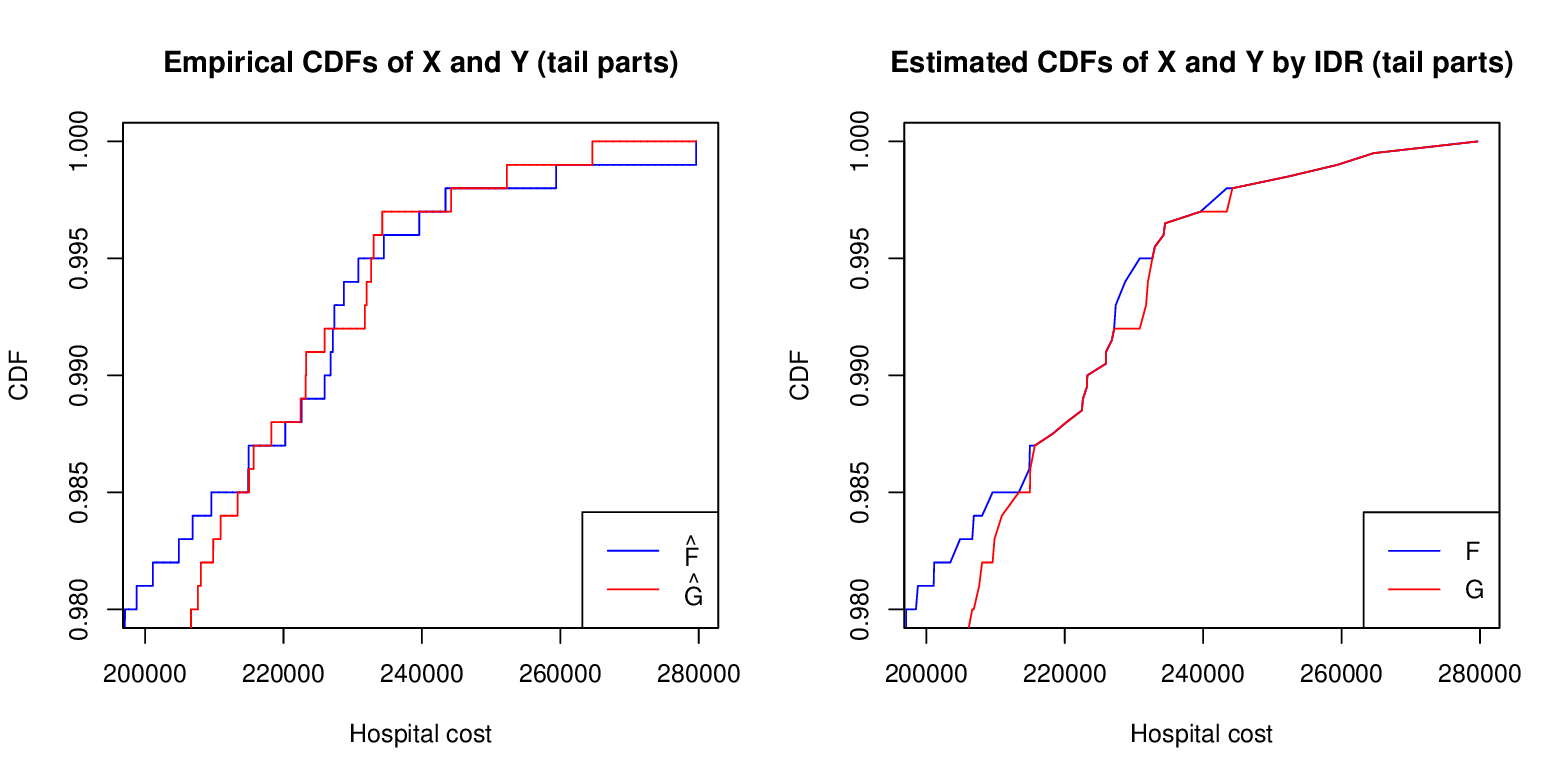}
\end{figure}

Using the IDR estimated distributions $F$ and $G$, VaR bounds and the improvements on the DU-spread in \eqref{eq:rw-april3}  are presented in Figure \ref{hosp}. VaR values are also presented if risks are independent, comonotonic, DL-coupled and countermonotonic. The extra order constraint greatly improves the unconstrained bounds of VaR, as shown by a  DU-spread reduction of more than $69\%$. In particular, the improvement on the best-case value is greater than that on the worst-case value.
  The reduction is almost $100\%$ when $p$ is close to $1$; that is because
  the two distributions $F^{[p,1]}$ and $G^{[p,1]} $ are almost identical for such $p$, making the set $\mathcal F^o_2 (F^{[p,1]},G^{[p,1]} )$ very small (see Figure \ref{cdf}, bottom panels).
Moreover, we observe that if the two risks $X$ and $Y$ are countermonotonic, VaR is close to the unconstrained lower bound. If the two risks are DL-coupled  or comonotonic, VaR is  close to the constrained lower bound.
\begin{figure}[htbp]
\centering
\caption{Case study: $\VaR_{p}$ bounds, DU reduction and $\VaR_{p}$ of the aggregate risk with different dependence structures are contained in this figure. VaR values with independence, comonotonicity, countermonotonicity and DL coupling are denoted by $\VaR^{\rm Ind}$, $\VaR^{\rm C}$, $\VaR^{\rm Co}$ and $\VaR^{\rm DL}$, respectively.}\label{hosp}
\includegraphics[height=6.5cm]{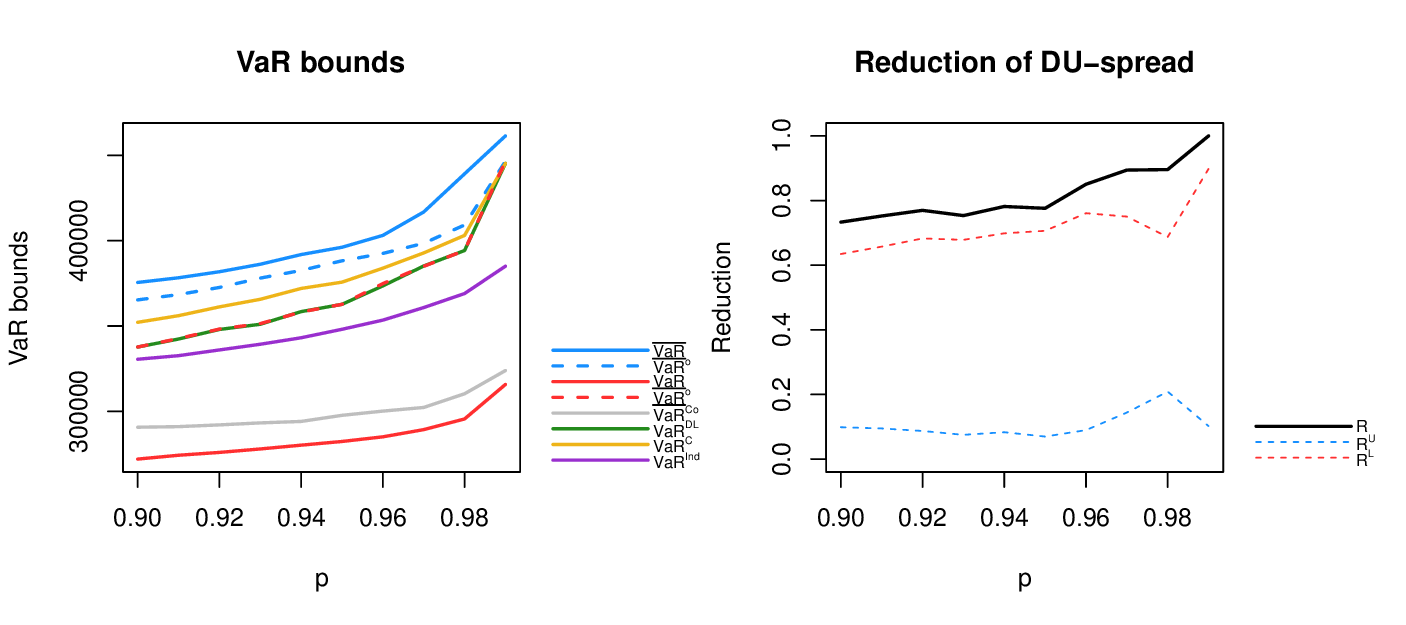}
\end{figure}

\section{Concluding remarks}\label{sec:con}

Risk aggregation of two ordered risks in the presence of unknown dependence structure is studied in this paper. The optimal dependence structures of the aggregate position are discussed in the sense of concave order, which can also be equivalently described via convex order. The largest (resp.~smallest) aggregate position in concave order is attained when the risks are  DL-coupled (resp.~comonotonic). The concave ordering bounds can be immediately applied to derive the bounds of $\le_{\rm cv}$-consistent and $\le_{\rm cx}$-consistent risk measures.

To analyze bounds on tail risk measures such as VaR,
 we introduce the notion of strong stochastic order and develop several theoretical properties. In particular, if the generator of the tail risk measure is $\le_{\rm cv}$-consistent, the worst-case value of the tail risk measure with the order constraint can be attained by $p$-concentrated risks, and it is attained when the upper-tail risks are DL-coupled. With a specific focus on VaR, analytical solutions are derived. Numerical studies show that the extra order constraint on top of the marginal distributions can significantly improve the bounds of risk measures which are solely based on marginal distributions.

 There are some limitations of the current setup considered in this paper. First, the assumption  $X\le Y$ for two risks $X$ and $Y$ is arguably quite  strong. As we have seen from the numerical results,  significant improvement of the constrained bounds over the unconstrained ones requires that the risks are of similar size, which however renders the ordering assumption difficult to satisfy.

We have focused on the problem of two ordered random variables in this paper, while a more general problem considering the order constraint among several risks in a large portfolio would also be interesting.
Such a constraint is motivated by monotone treatment effect analysis in causal inference (see \cite{M97}). The statistical inference of stochastically ordered distributions can be handled via IDR of \cite{henzi2019isotonic}.
 Let $G_1,\dots,G_n$ be  $n$ distributions satisfying $G_1\le_{\rm st} \dots \le_{\rm st}G_n$. Denote by
$$
\mathcal{R}^o _n=\{Y_1+\dots+Y_n: Y_i\sim G_i,~i=1,\dots,n,~Y_1\le\dots \le Y_n\}.
$$
We are interested in finding the worst-case value of a risk measure $\rho$ over the set $\mathcal{R}^o _n$.
If  $\rho$  is $\le_{\rm cx}$-consistent, then the worst-case value is attained by comonotonicity.
For $\rho$ that is not $\le_{\rm cx}$-consistent, such as the interesting case of VaR,
the problem is challenging and cannot be solved by the current techniques. Even without the order constraint, only limited analytical results are available for $n\ge 3$; see \cite{WPY13} and \cite{BLLW20}. We leave the theoretical analysis of this question, as well as the corresponding algorithms, for future work.

\subsection*{Acknowledgements}
R.~Wang
acknowledges financial support from the Natural Sciences and Engineering Research Council of Canada
(NSERC, RGPIN-2018-03823, RGPAS-2018-522590) and from the Center of Actuarial Excellence Research
Grant from the Society of Actuaries.

\appendix

\section{Proof of Lemma \ref{lem:conto}}
\begin{proof}[Proof of Lemma \ref{lem:conto}]
By Proposition 1 of \cite{embrechts2013note}, as $F$ and $G$ are strictly increasing and continuous, $F^{-1}$ and $G^{-1}$ are also strictly increasing and continuous.
We first show
\begin{align}\label{eq:rw-april1}
\lim_{\epsilon \downarrow 0}\overline{\VaR}_{p+\epsilon}^{R}(\mathcal F^o_2(F,G))= \overline{\VaR}_{p}^{R}(\mathcal F^o_2(F,G)).
\end{align}
By Theorem \ref{thm2},  for $\epsilon \ge [0,1-p)$,
 $$ \overline{\VaR}_{p+\epsilon}^{R}(\mathcal F^o_2(F,G)) = \essinf(X_\epsilon+Y_\epsilon) $$
 where   $(X_\epsilon, Y_\epsilon) \sim D_*^{F^{[p+\epsilon,1]},G^{[p+\epsilon,1]}} $.
Since $X_0$ and $Y_0$ have continuous distributions, using Corollary 2.5 of \cite{NW20}, we have $(X_\epsilon, Y_\epsilon)\to (X_0,Y_0) $ in distribution.
Since $\essinf$ is upper semicontinuous with respect to convergence in distribution,
we have $$\overline{\VaR}_{p}^{R}(\mathcal F^o_2(F,G))
= \essinf (X_0+Y_0) \ge \lim_{\epsilon \downarrow 0} \essinf(X_\epsilon+Y_\epsilon) = \lim_{\epsilon \downarrow 0}\overline{\VaR}_{p+\epsilon}^{R}(\mathcal F^o_2(F,G)), $$
which implies \eqref{eq:rw-april1}.
In what follows, we will show
\begin{align}\label{eq:rw-april2}
\lim_{\epsilon \downarrow 0}\overline{\VaR}_{p-\epsilon}^{R}(\mathcal F^o_2(F,G))= \overline{\VaR}_{p}^{R}(\mathcal F^o_2(F,G)).
\end{align}

Fix $p\in (0,1)$. If $F^{-1}(p)=G^{-1}(p)$, by Proposition \ref{cor:var}, $\overline{\VaR}_{p}^{R}(\mathcal F^o_2(F,G))=F^{-1}(p)+G^{-1}(p)$. For $\epsilon>0$, by Theorem \ref{thm2} and Corollary \ref{rm:cv},
$$\overline{\VaR}_{p-\epsilon}^{R}(\mathcal F^o_2(F,G))=\overline{\essinf}\left(\mathcal F^o_2 (F^{[p-\epsilon,1]},G^{[p-\epsilon,1]} )\right)\ge F^{-1}(p-\epsilon)+G^{-1}(p-\epsilon).$$
Thus we have
$$F^{-1}(p-\epsilon)+G^{-1}(p-\epsilon)\le\overline{\VaR}_{p-\epsilon}^{R}(\mathcal F^o_2(F,G))\le \overline{\VaR}_{p}^{R}(\mathcal F^o_2(F,G))= F^{-1}(p)+G^{-1}(p).$$
As $F^{-1}$ and $G^{-1}$ are continuous, let $\epsilon$ go to $0$, we get the desired result.

For the rest  of the proof, we assume  $F^{-1}(p)<G^{-1}(p)$. We first deal with the case where $F^{-1}(1)< \infty$ and $G^{-1}(1)<\infty$.
For $\epsilon>0$, let
$$\delta(\epsilon)=\sup_{p\le t\le 1}\left\{\left[F^{-1}(t)-F^{-1}(t-\epsilon)\right]\vee\left[G^{-1}(t)-G^{-1}(t-\epsilon)\right]\right\}.$$
As $F^{-1}$, $G^{-1}$ are continuous and $F^{-1}(1)< \infty$ and $G^{-1}(1)<\infty$, we have $0<\delta(\epsilon)<\infty$ and $\delta(\epsilon)\downarrow 0$ as $\epsilon \downarrow 0$.
Furthermore, let $$h(\epsilon)=\sup\left\{(F(G^{-1}(p))-p)-(F(z)-G(z)): z\in [G^{-1}(p-\epsilon),G^{-1}(p)]\right\}.$$
Because $F-G$ is continuous, we have $0\le h(\epsilon)<\infty$ and $h(\epsilon) \downarrow 0$ as $\epsilon \downarrow 0$.
 As $F^{-1}(p)<G^{-1}(p)$, we can take $\epsilon$ small enough such that $F(G^{-1}(p))-p>h(\epsilon)$ and $G^{-1}(p)-F^{-1}(p)>\delta(\epsilon)$.
By the definition of $\delta(\epsilon)$, we also have $$F^{-1}(p-\epsilon)<F^{-1}(p)<G^{-1}(p)-\delta(\epsilon)\le G^{-1}(p-\epsilon)<G^{-1}(p).$$
Define
$$x_\epsilon=\inf\left\{x: F(x)-(p-\epsilon)\ge F(G^{-1}(p))-p-h(\epsilon)\right\}.$$
 As $F$ is strictly increasing and continuous, we have $F^{-1}(p-\epsilon)<x_\epsilon< G^{-1}(p)$. Furthermore, $x_\epsilon \uparrow G^{-1}(p)$ as $\epsilon \downarrow 0$.
 Let $d(\epsilon)=G^{-1}(p)-x_{\epsilon}$. Thus, $0<d(\epsilon)<G^{-1}(p)-F^{-1}(p-\epsilon)$ and $d(\epsilon) \downarrow 0$ as $\epsilon \downarrow 0$. Furthermore, for any $x< x_\epsilon$, we have $F(x)-(p-\epsilon)<F(G^{-1}(p))-p-h(\epsilon)$.

From Proposition \ref{cor:var}, we have
$$\overline{\VaR}_{p-\epsilon}^R(\mathcal F^o_2(F,G))=\min\left\{\inf_{x\in\left[F^{-1}(p-\epsilon),G^{-1}(p-\epsilon)\right]}\left\{T^{F^{[p-\epsilon,1]},G^{[p-\epsilon,1]}}(x)+x\right\},2G^{-1}(p-\epsilon)\right\}.$$

\begin{enumerate}[(i)]

\item
For any $x \in \left[G^{-1}(p)-\delta(\epsilon)\vee d(\epsilon),G^{-1}(p-\epsilon)\right]$, we have
\begin{align}
T^{F^{[p-\epsilon,1]},G^{[p-\epsilon,1]}}(x)+x\ge 2x&\ge 2G^{-1}(p)-2\delta(\epsilon)\vee d(\epsilon) \notag
\\&\ge \overline{\VaR}_{p}^R(\mathcal F^o_2(F,G))-2\delta(\epsilon)\vee d(\epsilon).\label{eq:varcon1}
\end{align}
\item
 For any $x \in \left[F^{-1}(p-\epsilon),G^{-1}(p)-\delta(\epsilon)\vee d(\epsilon)\right)$, let $y=F^{-1}(F(x)+\epsilon)$. As $F(x)+\epsilon\ge p$, we have
\begin{equation}\label{eq:varcon2}
y-x=F^{-1}(F(x)+\epsilon)-F^{-1}(F(x))\le \delta(\epsilon),
\end{equation}
 and $y\le x+\delta(\epsilon)\le G^{-1}(p)$. Moreover,  we have $y\ge 
 F^{-1}(p)$. Therefore, $y\in \left[F^{-1}(p),G^{-1}(p)\right]$. By the definition of $h(\epsilon)$ and $x<x_\epsilon$, we have  for all $z \in [G^{-1}(p-\epsilon),G^{-1}(p)]$, $$F(z)-G(z)>F(G^{-1}(p))-p-h(\epsilon)>F(x)-(p-\epsilon).$$ Thus,
\begin{align*}
T^{F^{[p-\epsilon,1]},G^{[p-\epsilon,1]}}(x)&=\inf\{z> G^{-1}(p-\epsilon): F(z)-G(z)<F(x)-(p-\epsilon)\}\\
&=\inf\{ z> G^{-1}(p): F(z)-G(z)<F(y)-p\}\\
&=T^{F^{[p,1]},G^{[p,1]}}(y).
\end{align*}
 By \eqref{eq:varcon2}, we have
 \begin{equation}\label{eq:varcon4}
T^{F^{[p-\epsilon,1]},G^{[p-\epsilon,1]}}(x)+x=T^{F^{[p,1]},G^{[p,1]}}(y)+y+(x-y)\ge T^{F^{[p,1]},G^{[p,1]}}(y)+y-\delta(\epsilon).
\end{equation}

\end{enumerate}
Combining \eqref{eq:varcon1}, \eqref{eq:varcon4} and the fact that $G^{-1}(p-\epsilon)\ge G^{-1}(p)-\delta(\epsilon)\vee d(\epsilon)$, we conclude that, for $p \in (0,1)$,
 $$\overline{\VaR}_{p}^{R}(\mathcal F^o_2(F,G))-2\delta(\epsilon)\vee d(\epsilon)\le \overline{\VaR}_{p-\epsilon}^{R}(\mathcal F^o_2(F,G))\le \overline{\VaR}_{p}^{R}(\mathcal F^o_2(F,G)).$$
Letting $\epsilon \downarrow 0$, we get  \eqref{eq:rw-april2} for the case $F^{-1}(1)<\infty$ and $G^{-1}<\infty$.

If $F^{-1}(1)=\infty$ or $G^{-1}(1)=\infty$, following the proof of Proposition 4 in \cite{BLLW20}, we have $\overline{\VaR}_p^R(\mathcal F^o_2(F,G))=\overline{\VaR}_p^R\left(\mathcal F^o_2\left(F^{[0,m]},G^{[0,m]}\right)\right)$ for $p \in [0, 2m-1)$ and $ 1/2<m < 1$. Intuitively, extremely large values of risks do not contribute to the calculation of the worst-case value of VaR$^{R}$. As $\left(F^{[0,m]}\right)^{-1}(1)<\infty$ and $\left(G^{[0,m]}\right)^{-1}(1)<\infty$, $\overline{\VaR}_p^R\left(\mathcal F^o_2(F,G)\right)$ is continuous of $p \in (0,2m-1)$.
Letting $m \to 1$, we get  \eqref{eq:rw-april2}.\qedhere

\end{proof}

\end{document}